\documentclass[a4paper,11pt]{article}
\usepackage[margin=1in]{geometry}
\usepackage{url}

\usepackage{microtype}                 %
\PassOptionsToPackage{warn}{textcomp}  %
\usepackage{textcomp}                  %
\usepackage{mathptmx}                  %
\usepackage{times}                     %
\usepackage{cite}                      %

\usepackage[utf8]{inputenc}
\usepackage{amsmath}
\usepackage{amssymb}
\usepackage{amsthm}
\usepackage{paralist}
\usepackage{xspace}
\usepackage{color}
\usepackage[disable]{todonotes}
\usepackage[english]{babel}
\usepackage{booktabs}
\usepackage{subfig}

\newtheorem{assumption}{Assumption}
\newtheorem{lemma}{Lemma}

\newtheorem{theorem}{Theorem}
\newtheorem{observation}{Observation}

\DeclareMathOperator{\length}{length}

\newcommand{\ExternalLabeling}{\textsc{ExternalLabeling}\xspace}
\newcommand{\ContourLabeling}{\textsc{ContourLabeling}\xspace}

\newcommand{\IC}{\mathrm{C}}

\newcommand{\TB}{\mathrm{A}}
\newcommand{\TC}{\mathrm{B}}

\newcommand{\Sites}{\mathbb S}
\newcommand{\Contour}{\mathbb C}
\newcommand{\Figure}{\mathbb F}
\newcommand{\Ports}{\mathbb P}
\newcommand{\Instance}{\mathbb I}
\newcommand{\InstanceSet}{\textsc{Is}}

\newcommand{\OPTIONAL}[1]{}

\renewcommand{\restriction}{\raise-.2ex\hbox{\ensuremath|}}

\newcommand{\OPT}{\textsc{OPT}\xspace}
\newcommand{\TSCH}{\textsc{TSCH}\xspace}
\newcommand{\SCH}{\textsc{SCH}\xspace}
\newcommand{\CH}{\textsc{CH}\xspace}

\title{Radial Contour Labeling with Straight Leaders}

\author{Benjamin Niedermann\thanks{e-mail:niedermann@kit.edu}\\ %
        \scriptsize Karlsruhe Institute of Technology %
\and Martin Nöllenburg\thanks{e-mail:noellenburg@ac.tuwien.ac.at}\\ %
     \scriptsize  TU Wien %
\and Ignaz Rutter\thanks{e-mail:i.rutter@tue.nl}\\ %
   \scriptsize TU Eindhoven %
}
\date{}

\begin{document}

\maketitle

\begin{abstract}
  The usefulness of technical drawings as well as scientific
  illustrations such as medical drawings of human anatomy essentially
  depends on the placement of labels that describe all relevant parts
  of the figure. In order to not spoil or clutter the figure with
  text, the labels are often placed around the figure and are
  associated by thin connecting lines to their features,
  respectively. This labeling technique is known %
  as \emph{external label placement}.

  In this paper we introduce a flexible and general approach for
  external label placement assuming a \emph{contour} of the figure
  prescribing the possible positions %
  of the labels.  While much research on external label placement aims
  for fast labeling procedures for interactive systems, we focus on
  highest-quality
  illustrations. %
  Based on interviews with domain experts and a semi-automatic
  analysis of 202 handmade anatomical %
  drawings, we identify a set of 18 layout quality criteria, naturally
  not all of equal importance.  We design a new geometric label
  placement algorithm that is based only on the most important
  criteria. Yet, other criteria can flexibly be included in the
  algorithm, either as hard constraints not to be violated or as soft
  constraints whose violation is penalized by a general cost function.
  We formally prove that our approach yields labelings that satisfy
  all hard constraints and have minimum overall
  cost. %
  Introducing several speedup techniques, we further demonstrate how
  to deploy our approach in practice. In an experimental evaluation on
  real-world anatomical drawings we show that the resulting labelings
  are of high quality and can be produced in adequate time.
\end{abstract}

\section{Introduction}

Atlases of human anatomy play a major role in the education of medical
students and the teaching of medical terminology. Such books contain
a broad spectrum of filigree and detailed drawings of the human
anatomy from different cutaway views. For example, the third volume of
the popular human anatomy atlas Sobotta~\cite{sobotta} contains about
1200 figures on 384 pages. Figure~\ref{fig:example} is one of
them showing a cross section of the human skull. The
usefulness of the figures essentially relies on the naming of the
illustrated components. In order not to spoil the readability of the
figure by occluding it with text, the names are placed around the
figure without overlapping it. Thin black lines (which we call
\emph{leaders}) connecting the names with their features accordingly
guarantee that the reader can match names and features
correctly. Following preceding research, we call this labeling
technique \emph{external label placement}.  In this paper we present a
flexible and versatile approach for external label placement in
figures.  We use medical drawings as running example, but occlusion-free label placements are also indispensable for the readability of
other highly detailed figures as they occur for example in scientific
publications, mechanical engineering and maintenance manuals.

Besides readability also the aesthetics of the figures including
their labelings play a central role in professional books. Each
figure and its labels are subject to book-specific design rules.  Our
approach stands out by its ability to support an easy integration of
these specific design rules. It particularly relies only on a few key
assumptions that most figures with external label placement have in
common. Other constraints and rules can easily be patched in according
to demand.

\begin{figure}[t]
  \centering
  \includegraphics[width=0.8\linewidth]{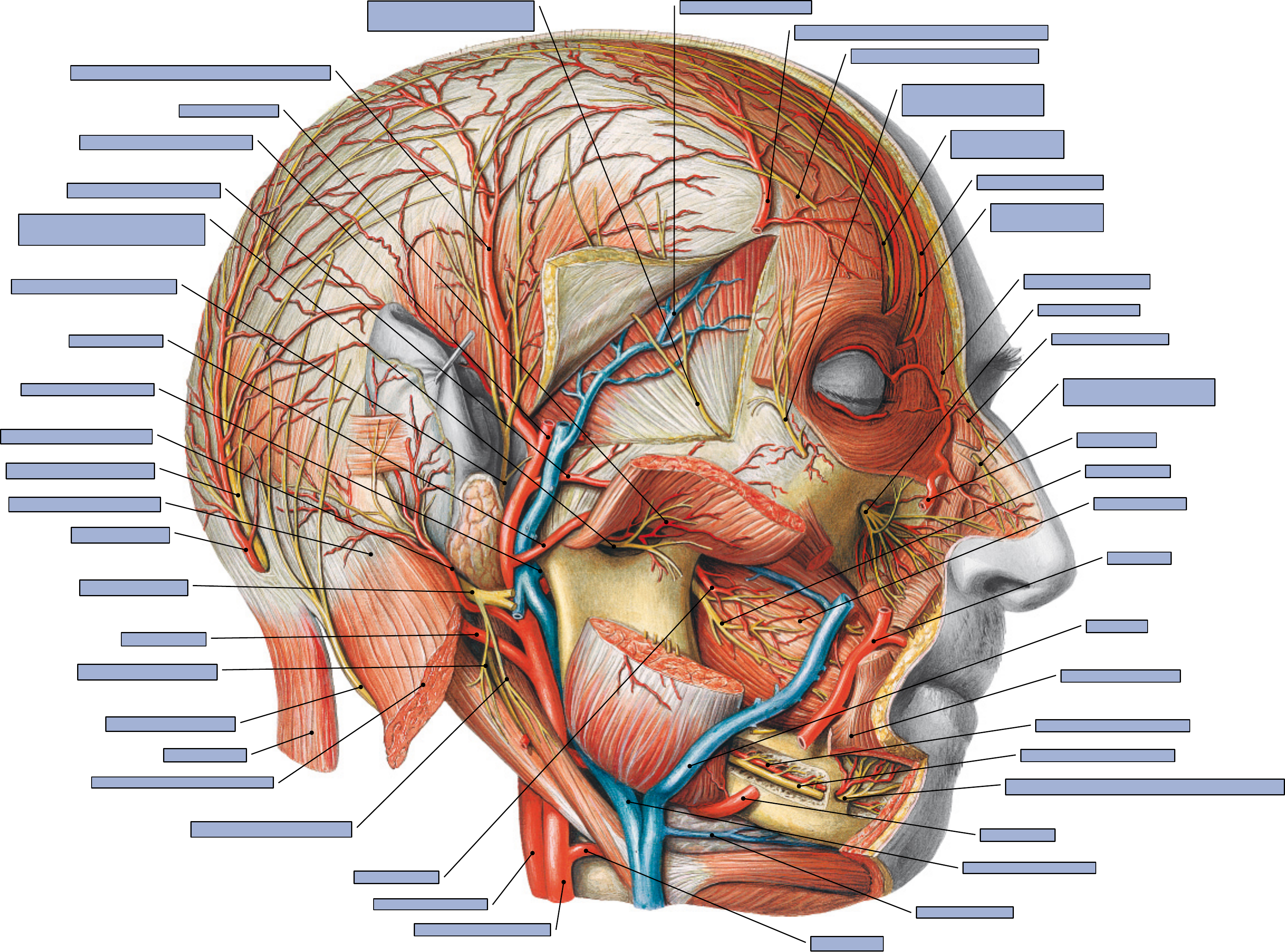}
  \caption{Drawing labeled by our approach (variant
    \TSCH, 34 sec.). Source: Paulsen, Waschke, Sobotta Atlas
    Anatomie des Menschen, 23.Auflage 2010 \copyright Elsevier GmbH, Urban \&
    Fischer, München. }
  \label{fig:example}
\end{figure}

To validate our approach, we were in contact with both a layout artist
and two editors of the human anatomy atlas Sobotta. Both the layout
artist and the two editors stated that label placement is a
mechanical, but extensively time-consuming task that is mainly done by
hand. The tool support basically comprises simple operations such as
placing text boxes and drawing line segments.  Based on medical
drawings annotated by the authors of the atlas for human anatomy, the
layout artist creates the layout of each double page of the
book. Using the annotated information, this includes arranging
explanatory texts, figures, and labels around the figures. The
interviewed layout artist stated that he needs about two hours to
create the layout of a double page. Hereby, he spends a large portion
of his time on label placement. Hence, with around 1200 figures in a
single volume better tool support would clearly help in improving the
process of creating such books. Further, with the upcoming
applications on mobile devices, figures are deployed in
differently scaled settings, which requires different external
labelings for the same figure. Then, at the latest, automatic approaches
become inevitable.

\textbf{Related Work.} External labeling algorithms have been investigated both
from a practical and theoretical point of view. The practical results
aim for fast and simple approaches that support heuristic
optimizations for multiple criteria.  These approaches are often
targeted for interactive systems requiring a fast label placement
procedure making compromises concerning quality.  For example Hartmann
et al.~\cite{Hartmann2004} and Ali et al.~\cite{Ali2005} both present
models for external label placement listing a set of criteria
concerning readability and aesthetics. They use these models to
introduce simple force-based methods. Further strategies comprise
iteratively rearranging labels~\cite{Bruckner2005}, sequentially
evaluating and placing labels by priority~\cite{Gotzelmann2006} or
exploiting small spaces for label placement~\cite{Fuchs2006}.
These approaches have in common that they are fast and simple,
  but quality guarantees can not be given. Hence, while they are
  customized to suit interactive systems, they can hardly serve as
  tools for a designer of professional books demanding complex design
  rules.   Alternatively, the drawing criteria are transferred into
an optimization problem based on a sophisticated mathematical model.
Typically, such problems are not solved exactly, but by local
optimization approaches. For example Vollick et al.~\cite{Vollick2007}
model external label placement by means of energy functions and apply
simulated annealing. Stein and D\'ecoret~\cite{Stein2008} use a simple
greedy algorithm to find a solution for mathematical
constraints. Further, Čmolík and Bittner \cite{Cmolik2010} employ a
greedy optimization using a model based on fuzzy logic. Mogalle et
al.~\cite{Mogalle2012} also applied greedy optimization. The
applications range from the interactive exploration of volume
illustrations~\cite{Bruckner2005}, over automatically annotated 2D
slices of segmented structures~\cite{Muhler2009}, up to labeling
explosion diagrams in augmented
reality~\cite{Tatzgern2013,Tatzgern2014}.

In contrast, theoretical results mostly consider simple models,
typically with one optimization criterion, e.g., minimizing the total
leader length. Bekos et al.~\cite{Bekos2007} introduced the first such
model. Typically, it is assumed that the point features to be labeled
are contained in a rectangle~$R$ representing the boundary of
the figure. The labels are placed outside that rectangle touching
  the boundary. It is assumed that the labels have uniform shapes and
  that their bounding boxes are already placed alongside~$R$.
  Hence, the
  labeling problem basically becomes a geometric matching problem
  asking for a crossing-free assignment between the placed bounding
  boxes and the point features such that each point feature is
  connected with a unique bounding box by a leader. In a
  post-processing step the label texts are written into the bounding
  boxes, accordingly. This sub-problem of external labeling is
called \emph{boundary labeling}. The preceding research mostly differs
in the choice of the parameters. Typically, the type of the leaders is
broken down into straight-line
leaders~\cite{Bekos2007,Gemsa2015b,Fink2012}, L-shaped
leaders~\cite{Bekos2007,Benkert2009,Noellenburg2010,Kindermann2015},
S-shaped leaders~\cite{Bekos2007,Bekos2010b,Huang2014} and leaders
with a diagonal segment~\cite{Benkert2009,Bekos2010}. Further, the
results distinguish between the number of sides on which the labels
may lie; on one, two or multiple sides of~$R$.  Typical
optimization functions are minimizing the total leader length or
minimizing the number of bends. Some of the approaches also allow
general cost functions (e.g.,~\cite{Benkert2009}) applying dynamic
programming. Recently, Fink and Suri~\cite{Fink2016}
  presented dynamic programming approaches for boundary labeling with
  uniform labels and rectangular boundaries using the four major
  leader types, which include straight-line leaders. For labels on two
  opposite sides of the rectangle, their approaches run in $O(n^{15})$
  (fixed label candidates) and $O(n^{27})$ (non-fixed label
  candidates) time for straight-line leaders.  A more detailed survey
was recently given by Barth et al.~\cite{Barth2015}, which also
comprises a user-study showing that straight-line leaders and L-shaped
leaders outperform the other mentioned leader types concerning
readability.
While the results on boundary labeling are interesting from a
theoretical perspective, they are hardly applicable in realistic
settings, where labels are not uniform and the boundary of the figure
is not a rectangle. 

\textbf{Our Contribution.} Our approach bridges the gap between the
practical and theoretical results. Like many of the theoretical
results, it uses a clearly and mathematically defined model to
guarantee pre-defined design rules. However, in contrast to preceding
research our approach is significantly more flexible in adapting other constraints.  After
introducing some terminology (Sect.~\ref{sec:terminology}) we present
an extensive list of drawing criteria for external label placement. 
Our list of criteria contains and refines the quality criteria listed in previous work~\cite{Hartmann2004,Ali2005} and identifies several more.
We stress that, in contrast to previous work, our criteria are obtained directly from interviews with layout artists and editors of anatomical atlases. Moreover, we empirically verified the level of compliance with these criteria for 202 figures printed in the
Sobotta~\cite{sobotta} atlas (Sect.~\ref{sec:criteria}) using a semi-automatic quantitative analysis.

Based on a reasonable subset of the most important criteria we
introduce a flexible formal model for \emph{contour labeling}, which
is a generalization of boundary labeling
(Sect.~\ref{sec:model}). Afterwards we describe a basic dynamic
programming approach that solves the mathematical problem optimally
(Sect.~\ref{sec:dp}).  Our approach allows to include further drawing
criteria both as \emph{hard} and \emph{soft} constraints, where hard
constraints may not be violated at all and the compliance of soft
constraints is rated by a cost function.  Previous work rarely uses
hard constraints and cannot easily include new hard constraints.
Moreover, in contrast to previous algorithms that compute mathematical optimal solutions, our approach also takes
consecutively placed labels into account. At first glance this seems
to be a small improvement, but in fact it is important to obtain an
appealing labeling where, for example, labels have regular distances
or the angles of consecutive labels should be similar.  Further, our
approach supports labels of different sizes. Indeed, for each point
feature, the user can pre-define a set of different label shapes,
which do not need to be rectangles. This may be used to model
different ways of text formatting supporting single and multi-line
labels. Further, the user may specify for each label an individual set
of candidate positions that is used for the label placement
procedure. Moreover, the user may mark areas that are not allowed to
be overlapped by labels including their leaders. This is important to
avoid undesired overlaps with the figure or to integrate the figure
along with its labeling into a double page with explanatory text. The
approach also allows to pre-define groups of labels that are placed
consecutively, which is required when naming features that are
semantically related.

Our approach is not limited to the described features, but other
criteria can be incorporated easily. The strength of our approach
comes at the cost of a high asymptotic running time $O(n^8)$,
where $n$ describes the complexity of the input
instance. Recently, Keil et al.~\cite{Keil2016} presented a
  similar general dynamic programming approach for computing an
  independent set in outerstring graphs, which can be utilized to
  solve contour labeling in $O(n^6)$ time for a general cost
  function rating individual labels; however, it cannot take joint costs of two consecutive labels into account. 
In contrast to Fink and
  Suri~\cite{Fink2016} our approach is significantly faster ($O(n^8)$
  instead of $O(n^{15})$) and it supports non-uniform labels and
  more general shapes.  
With some engineering
  (Sect.~\ref{sec:engineering}) we can solve realistically sized
  instances in adequate time and high layout quality as is shown in our evaluation (Sect.~\ref{sec:experiments}) on a large benchmark set of real-world instances.

\section{Terminology}\label{sec:terminology}
\begin{figure}[t]
  \centering
  \includegraphics{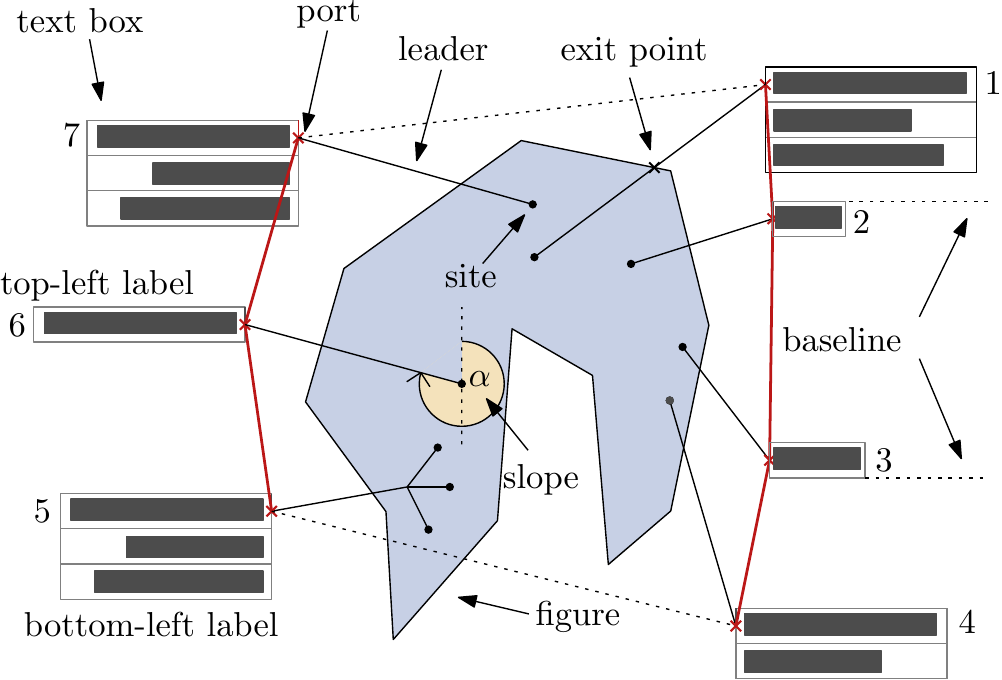}
  \caption{Terminology. The radial ordering is given by the numbers
    placed closely to the text boxes.}
  \label{fig:terminology}
\end{figure}
An \emph{illustration} with \emph{external labeling} consists of a
\emph{figure} as well as a set of labels outside of the figure naming
single point features of the figure.  Hereby a \emph{label} consists
of a \emph{text box} that lies outside of the figure and a line
segment connecting the text box with its point feature; see
Fig.~\ref{fig:terminology}. We call the line segment the \emph{leader}
and the point feature the
\emph{site} %
of the label.
More precisely, the text box is a rectangle containing a (possibly
multi-line) text. Typically, the rectangle is not displayed, but it is
used for the further description. We assume that the leader of a label
ends on the boundary  of the text box; we call that
point the \emph{port} of the label. A leader is directed from its site
to its port.  A label whose leader goes to the left is called a
\emph{left label} and a label whose leader goes to the right is called
a \emph{right label}. Analogously, a label is a \emph{top label}
(\emph{bottom label}) if its leader goes upwards (downwards).  The
\emph{baseline} of a bottom-right label is the horizontal half line
that emanates from the bottom-right corner of the label's text box to
the right. For a top-right label the \emph{baseline} is the horizontal
half line that emanates from the top-right corner of the label's text
box to the right. The baselines for bottom-left and top-left labels are
defined symmetrically.  We define the \emph{slope} of a leader is the
clockwise angle (starting at 12 o'clock) between the leader and the
vertical segment going through the connected site.  A leader of a
labeling intersects the contour of the figure at its \emph{exit
  point}; in case that a leader intersects a figure multiple times, we
regard the intersection point closest to the port. Traversing the
  figure's boundary in clockwise order starting from the boundary's
  topmost point  defines an ordering on the
  exit points of the leaders and accordingly on the labels; we call
  this the \emph{radial ordering} of the labeling.  Two labels are
\emph{consecutive} if one directly follows the other in the radial
ordering.  The \emph{labeling contour} is the polygon
 that connects the ports of the labels in the given
radial ordering.

\section{Drawing Criteria}\label{sec:criteria}

\begin{figure}[t]
  \centering
  \includegraphics[page=2,width=\linewidth]{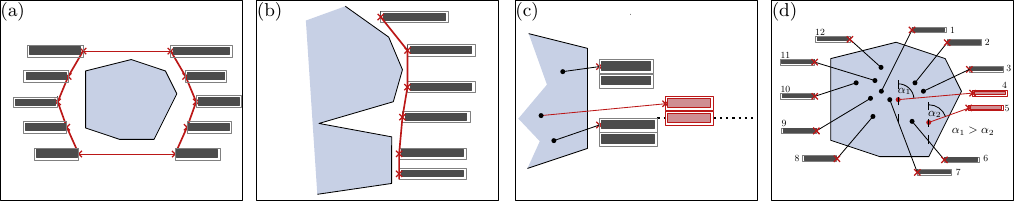}
  \caption{Drawing criteria. (a)~Criterion~G\ref{crit:shape}
    (b)~Criterion~\ref{crit:similarity}
    (c)~Criterion~T\ref{crit:shadowing}, which is violated by the red
    middle label.  (d)~Criterion~L\ref{crit:monotone}, which is violated by Label 4 and Label 5.}
  \label{fig:criteria}
\end{figure}

We conducted interviews with domain experts (one layout artist and two
editors of~\cite{sobotta}) in order to extract a comprehensive set of
important layout quality criteria as listed below. During these
interviews, the experts explained how they typically proceed when
labeling a single illustration. Further, they were asked to list
layout criteria that they explicitly take into account. The extracted
formal criteria were subsequently discussed using example
illustrations printed in Sobotta~\cite{sobotta}.

\noindent\textbf{Drawing criteria for sites.} 
\begin{compactenum}[S1]
\item  \textit{Shape.} Sites are represented by small points in the
  drawing. 
\item \textit{Position.} \label{crit:sites}The position of a site is prescribed by domain
  experts and can be assumed to be fixed and
  given with the input.
\item \textit{Type.} %
  Either a site has its own label, or
  multiple sites have the same label. In the latter case the leaders
  are bundled forking at a certain point; see Label 5 in
  Fig.~\ref{fig:terminology}.
\end{compactenum}
For simplicity, we assume that we only have sites of the
first type, but with some engineering our algorithms can also be
adapted to the second case.  Now consider an external
labeling of a medical drawing. We have extracted the following
criteria.

\noindent \textbf{General drawing criteria.}
\begin{compactenum}[G1]
\item \label{crit:extern}\textit{Externality.} The text boxes
  of the labels are placed outside the drawing in the available
  areas.
\item \label{crit:planar}\textit{Planarity.}  To sustain
  readability, labels may not overlap or intersect each other.
\item \label{crit:shape}\textit{Simple shape.} The labeling
  contour should be simple avoiding turning points; see
  Fig.~\ref{fig:criteria}(a). 
\item \label{crit:sides}\textit{Left/right side.} The radial order of a labeling can be
  partitioned into a sequence of left labels and a sequence of right
  labels. Consequently, the labeling contour can be partitioned into a
  \emph{left labeling contour} and a \emph{right labeling contour}.
\item \textit{Similarity.}\label{crit:similarity} The labeling contour
  \emph{mimics} the contour of the figure such that small
  ``indentations'' of the figure are not taken into account; see
  Fig.~\ref{fig:criteria}(b).
\item \textit{Grouping.} Labels may be required by the designer to
  appear consecutively in the radial ordering of the labeling.
\end{compactenum}

\noindent \textbf{Drawing criteria for text boxes.}
\begin{compactenum}[T1]
\item \textit{Spacing.}  The vertical distances between text boxes are
  preferably uniform. Distances less than the height of one text line
  should be avoided, but may be admissible if not preventable.
\item \textit{Appearance.} For all labels the same font is applied. Differences in the importance of labels may be expressed by different emphasis (bold, italic).
\item \textit{Single/Multi line.} If possible, a text box should
  consist of a single-line text. Only due to the available space,
  text boxes may consist of multiple lines.
\item \label{crit:ports}\textit{Ports.} For left (right) labels the port lies on the right (left) edge of the text box in the vertical center of the first text line. 

\item \label{crit:shadowing}\textit{Staircase.} Let~$\ell_1$
  and $\ell_2$ be two consecutive labels. Neither~$\ell_1$ nor
  $\ell_2$ intersects the baseline of the other; see
  Fig.~\ref{fig:criteria}(c).
\end{compactenum}

\noindent \textbf{Drawing criteria for leaders.}
\begin{compactenum}[L1]
\item \label{crit:length}\textit{Length}. The part of a leader
  covering the figure should be (preferably) short.
\item \textit{Distinctiveness.} Leaders running close together
    should not be parallel to avoid reader confusion.
\item \textit{Distance.} Leaders preferably comply with a minimum
  distance to sites of other leaders.
\item \label{crit:monotone}\textit{Monotonicity.} The slope of the
  leaders increases with respect to the radial ordering of the
  leaders; see Fig.\ref{fig:criteria}(d).
\end{compactenum}

Typically a labeling does not fully satisfy all these criteria, but
criteria may contradict each other requiring appropriate
comprises. For example requiring monotonicity (L\ref{crit:monotone})
may enlarge the total leader length, which conflicts with criterion
L\ref{crit:length}.  Our approach is characterized by the fact that
these compromises are not already made during the design of the
labeling algorithm, but they lie in the hand of the layout artist
applying the algorithm. Specifically, our approach only needs criteria
S\ref{crit:sites}, G\ref{crit:extern}, G\ref{crit:planar},
G\ref{crit:sides} and T\ref{crit:shadowing} as hard constraints not to
be violated. Further, we assume that we are given a simple shape
(G\ref{crit:shape}) enclosing the figure and prescribing possible
positions of ports. All other criteria are optional, but can be easily
patched in as either hard or soft constraints as needed. Hereby the
compliance of soft constraints is rated by means of a general cost
function that can be defined when applying the algorithm.  In our
interviews the domain experts strongly emphasized the importance of
G\ref{crit:planar} and G\ref{crit:shape}. They further pointed out
that labels should not be placed behind other labels, which we express
by T\ref{crit:shadowing}.  We further analyzed 202 medical drawings of
Sobotta~\cite{sobotta} in a semi-automatic way.  To that end we
  vectorized the images by extracting the text boxes, the leaders, the
  sites and the contour of the figure. More precisely, we computed the
  difference of two images showing the same object; one with labels
  and one without labels. From this difference we automatically
  extracted the contour of the figure, the text boxes and the ports of
  the leaders. The leaders were manually extracted. In case that a
leader was connected to multiple sites, we have pragmatically
extracted the leader only up to the first fork point~$p$ and placed a
single site at this point.  In case that~$p$ was not
contained~$\Figure$, we took the projection of~$p$ on the boundary of
$\Figure$ along the half-line from the port of $\lambda$ through $p$.
The figures contain between $4$ and $64$ sites; see also
Fig.~\ref{fig:quality:pd10}. 

 All of these examples satisfy
S1, G1, G2, G4, and T4. Further, 18 figures contain at least one set
of labels that are explicitly grouped by a large curly brace (G6).
Only a dwindling small percentage $(0.4\%)$ of all labels violate the
staircase property (T5) and about $6.2\%$ violate monotonicity
(L4). Since the other criteria are soft, we did not quantitatively
check these in the semi-automatic analysis; yet, they are well founded
in the conducted interviews.

\section{Formal Model}\label{sec:model}
We now describe a formal model for external label placement.  We are
given a simple polygon~$\Figure$ that describes the contour of the
figure and contains $n$ sites to be labeled. We denote the set of the
sites by $\Sites$ and assume that the sites are in general position,
i.e., no three sites are colinear\footnote{This assumption can be met
  by slightly perturbing sites.}.  For each site~$s\in \Sites$ we
describe its \emph{label}\footnote{To ease presentation we define that the leader is a
  component of the label. In preceding research only the rectangle~$r$
  is called label.}~$\ell$ by a rectangle~$r$ and an oriented
line segment $\lambda$ that starts at $s$ and ends on the boundary
of~$r$.  We call $\lambda$ the \emph{leader} of $\ell$,
$r$ the \emph{text box} of $\ell$, and the end point of~$\lambda$
on~$r$ the \emph{port} of $\ell$.  The other end point is the site
of~$\ell$; see Fig.~\ref{fig:terminology}. In the following we only
consider labels whose text boxes satisfy T4.

A set~$\mathcal L$ of labels over~$\Sites$ is called an \emph{external
  labeling} of $(\Figure,\Sites)$, if
\begin{inparaenum}[(1)]
\item $|\mathcal L|=|\Sites|$,
\item for each site $s\in \Sites$ there is exactly one label in~$\mathcal L$ that belongs to $s$, and 
\item every text box of a label in~$\mathcal L$ lies outside
  of~$\Figure$. 
\end{inparaenum} If no two labels in $\mathcal
L$ intersect each other, $\mathcal L$ is \emph{planar} . A labeling~$\mathcal L$ is called a
\emph{staircase labeling} if it satisfies criterion~T\ref{crit:shadowing}.

Let~$\mathcal L$ be a planar labeling. Let $\ell_1,\dots,\ell_n$ be
the labels of~$\mathcal L$ in the radial ordering. For simplicity we
define~$\ell_{n+1}:=\ell_{1}$. The \emph{cost}~$c$ of a
labeling~$\mathcal L$ is defined as
$c(\mathcal L)=\sum_{i=1}^n c_1(\ell_i) + c_2(\ell_i,\ell_{i+1})$, 
where $c_1$ is a function assigning a cost to a single label~$\ell_i$
and $c_2$ is a function assigning a cost to two consecutive
labels~$\ell_i$ and $\ell_{i+1}$. We note that in contrast to previous
research the cost function also supports rating two consecutive
labels, which is crucial to set labels in relation with each other.
Given the cost function~$c$, the problem \ExternalLabeling then asks
for a planar labeling~$\mathcal L$ of $(\Figure,\Sites)$ that has
minimum cost with respect to $c$, i.e., for any other planar
labeling~$\mathcal L'$ of $(\Figure,\Sites)$ it holds~$c(\mathcal
L)\leq c(\mathcal L')$.

We consider the special case that the ports of the
  labels lie on a common \emph{contour} enclosing $\Figure$.  In contrast to
  classical boundary labeling, which assumes a rectangular figure,
  this contour schematizes the shape of the figure with a certain
  offset; in Section~\ref{sec:experiments} we shortly describe how to
  construct a reasonable contour. Thus, the contour describes the
  common silhouette formed by the labels.  We assume that the contour is
  given as a simple polygon~$\Contour$
  enclosing~$\Figure$. An external labeling~$\mathcal L$ is called a \emph{contour
    labeling} if for every label of $\mathcal L$ its leader lies
  inside~$\Contour$ and its port lies on the boundary~$\partial\Contour$ of~$\Contour$.
  Since not every part of $\Contour$'s boundary may be suitable for the
  placement of labels, we require that the ports of the labels are
  contained in a given subset $\Ports\subseteq \partial\Contour$ of candidate
  ports. If $\Ports$ is finite, the input instance has
  \emph{fixed ports} and otherwise \emph{sliding ports}.

\begin{observation}
  In a planar contour labeling the ports of the labels induce the
  same radial ordering with respect to $\Contour$ as the exit
  points of the labels with respect to $\Figure$.
\end{observation}

A tuple $\Instance=(\Contour,\Sites,\Ports)$ is called an
\emph{instance} of contour labeling. The \emph{region} of $\Instance$
is the region enclosed by $\Contour$.  We restrict ourselves to convex
contours and clearly separated sites and text boxes as follows implementing
Criteria~G\ref{crit:extern} and~G\ref{crit:shape}, respectively.

\begin{assumption}\label{assumption:convex-hull}
  The contour $\Contour$ is convex and no text box of any label 
  intersects the convex hull of $\Sites$.
\end{assumption}

For all of the 202 analyzed medical drawings it holds that no text box
of any label intersects the convex hull of~$\Sites$. 

Due to the convexity of~$\Contour$, the contour can be uniquely split
into a left and right side described by two maximal $y$-monotone
chains~$\Contour_\mathrm{L}$ and $\Contour_\mathrm{R}$,
respectively. The next assumption implements Criterion~G4.

\begin{assumption}\label{assumption:independent-chains}
  A left label has its port on $\Contour_\mathrm{L}$ and a right
    label has its port on $\Contour_\mathrm{R}$.
\end{assumption}

Given a cost function~$c$, the problem \ContourLabeling  then asks for
an (cost) optimal, planar staircase contour labeling~$\mathcal L$ of
$(\Contour,\Sites,\Ports)$ with respect to $c$, i.e., for any other
planar staircase contour labeling~$\mathcal L'$ of
$(\Contour,\Sites,\Ports)$ it holds that~$c(\mathcal L)\leq c(\mathcal
L')$.

\section{Algorithmic Core}\label{sec:dp}

In this section we describe how to construct the optimal
labeling~$\mathcal L$ of a given instance $(\Contour,\Sites,\Ports)$
with respect to a given cost function~$c$. To that end we apply a
dynamic programming approach. The basic idea is that any optimal
contour labeling can be recursively decomposed into a set of
sub-labelings inducing disjoint sub-instances. As we show later, these
sub-instances are specially formed; we call them \emph{convex sub-instances}.  We
further show that any such sub-instance can be described by a constant
number of parameters over $\Sites$ and $\Ports$. Hence, enumerating
all choices of these parameters, we can enumerate in polynomial time
all possible convex sub-instances that an optimal labeling may consist
of. For each such sub-instance we compute the cost of an optimal
labeling reusing the results of already computed values of smaller
sub-instances. In this way we obtain the value of the optimal labeling
for the given instance. Summarizing, our approach consists of four
steps.

\textsc{Step 1}. Compute all possible convex sub-instances by
enumerating all possible choices defined  over $\Sites$ and
$\Ports$.

\textsc{Step 2}. In increasing order of the number of contained sites,
compute the optimal cost for each convex
sub-instance~$\Instance$. More precisely, to compute the optimal cost
of~$\Instance$ consider all possibilities how $\Instance$ can be
composed of at most two smaller convex sub-instances.

\textsc{Step 3}. Consider all possibilities how the input instance can
be described by a convex sub-instance. Among these, take the convex
sub-instance with optimal cost.

\textsc{Step 4}. Starting with the resulting sub-instance of
\textsc{Step 3}, apply a standard backtracking approach for dynamic
programming to construct the corresponding labeling with optimal
costs.

In the remainder of this section we explain the approach in more
detail. In Section~\ref{sec:struc-properties}, we first prove some
structural properties on contour labelings. These properties are
crucial for the dynamic programming approach, which we describe in
Section~\ref{sec:dyn-prog}.

\subsection{Structural Properties of Contour Labelings}
\label{sec:struc-properties}
The intersection of two labels is characterized by three types: the two leaders intersect, the two text boxes intersect or the leader of one label intersects the text box of the other label.
The following lemma allows us to find planar labelings by avoiding
leader-leader intersections and intersections between two
consecutive labels.

\begin{figure}[t]
\centering
   \includegraphics[page=4]{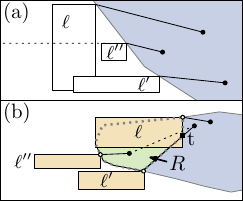}
   \caption{   Illustration of proof for Lemma~\ref{lem:shadow-free}.}\label{fig:shadowed-labels}
\end{figure}

\begin{lemma}\label{lem:shadow-free}
  Let $\Instance$ be an instance of \ContourLabeling and let~$\mathcal
  L$ be a staircase contour labeling of $\Instance$. If no pair of leaders
  intersect and if no two consecutive labels intersect, then $\mathcal
  L$ is planar.
\end{lemma}

\begin{proof}
   We prove that $\mathcal L$~is planar by systematically
    excluding the possible types of intersections.

  \textit{Text-box--text-box intersection.} Assume that there are two
  labels $\ell$ and $\ell'$ that are not consecutive and whose text
  boxes intersect. The labels either lie on the same side or on
  different sides of~$\Contour$.
  
  First consider the case that $\ell$ and $\ell'$ belong to different
  sides; without loss of generality let $\ell$ be a left label and
  $\ell'$ be a right label. Due to T4, text boxes of $\ell$ and
  $\ell'$ may only intersect if the port~$p$ of $\ell$ lies to the
  right of the port~$p'$ of $\ell'$. Since~$p$ lies on
  $\Contour_\mathrm{L}$ and $p'$ lies on $\Contour_\mathrm{R}$, this
  contradicts the convexity of $\Contour$.

  Now consider the case that $\ell$ and $\ell'$ belong to the
  same side; without loss of generality let both be left labels; see
  Fig.~\ref{fig:shadowed-labels}(a). For intersecting each other, one of
  both labels intersects the base line of a left label in between both
  labels contradicting Criterion~T5.
  
 \textit{Text-box--leader intersection.}  Now assume that there is
  a label $\ell$ whose text box is intersected by the
  leader~$\lambda'$ of another label~$\ell'$; see
  Fig.~\ref{fig:shadowed-labels}(b). We denote the ports of $\ell$ and
  $\ell'$ by $p$ and $p'$ respectively.  Further, let~$t$ be the first
  intersection point of $\lambda'$ with $\ell$ going along
  $\lambda'$. We choose $\ell'$ such that there is no other leader
  intersecting~$\ell$'s boundary~$c$ between $p$ and $t$. Let~$R$ be
  the region that is bounded by the boundary of $\ell$ from $p$ to
  $t$, the line segment $\overline{tp'}$, and the boundary $c'$ of
  $\Contour$ from $p'$ to $p$.  Since $\ell$ and $\ell'$ are not
  consecutive, there is a label $\ell''$ with port $p''$ on~$c'$. The
  site $s''$ of $\ell''$ lies in $R$ because otherwise the leader of
  $\ell''$ intersects $c$ or the segment~$\overline{tp'}$. Due to the
  convexity of~$\Contour$, the segment $\overline{s's''}$ is contained
  in~$\Contour$. Since~$s'$ lies in the complement of~$R$, the segment
  $\overline{s's''}$ intersects $c$, which implies that $\ell$
  intersects the convex hull of $\Sites$ contradicting
  Assumption~1.
\end{proof}

\begin{figure}[t]
 \centering
 \includegraphics[width=\linewidth]{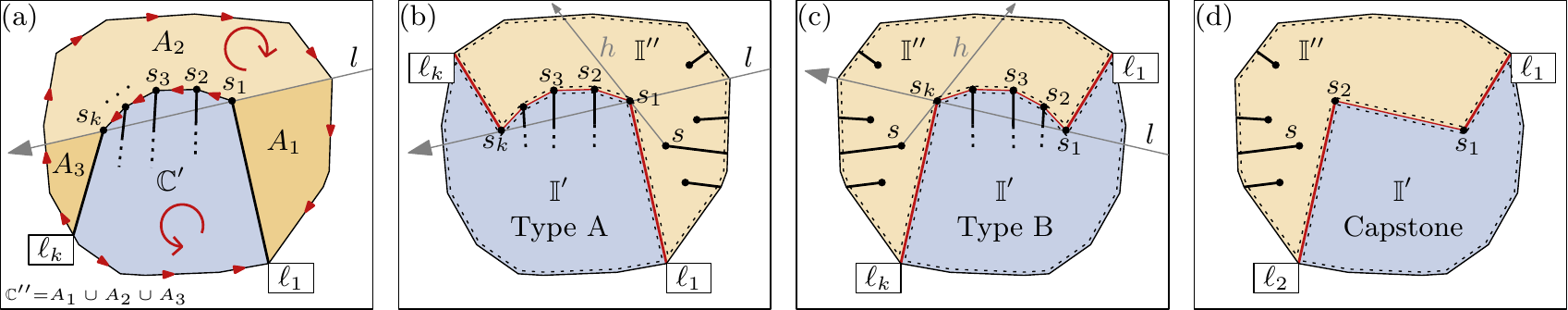}
 \caption{Decomposition in convex instance $\Instance'$ (blue) and concave instance $\Instance''$ (orange). (a)~Illustration of basic definitions. (b)~Type~$A$ instance with $k>2$. (c)~Type $B$ instance with $k>2$. (d)~Capstone instance. }
 \label{fig:basic-instances}
\end{figure}

In the following, we show that each instance can be subdivided into a
finite set of sub-instances of three types.  We describe a
sub-instance by a simple polygon that consists of two polylines. One
polyline is part of the original contour $\Contour$ and the other
polyline consists of a convex chain of sites and two leaders.  More
precisely, assume that we are given a convex chain
$K=(s_1,\ldots,s_k)$ of sites with $k\geq 2$ and the two
non-intersecting labels $\ell_1$ and $\ell_k$ of $s_1$ and $s_k$,
respectively; see Fig.~\ref{fig:basic-instances}(a).  The directed
polyline $K'=(p_1,s_1,\dots,s_k,p_k)$ splits the polygon $\Contour$
into two polygons $\Contour'$ and $\Contour''$, where $p_1$ and
$p_k$ are the ports of $\ell_1$ and $\ell_k$, respectively. We consider the order of
the sites such that we meet $p_1$ before $p_k$ when going along the
contour of $\Contour$ in clockwise-order starting at the top of
$\Contour$. Further, going along $K'$ we denote the sub-polygon to the
left of $K'$ by $\Contour'$ and to the right of $K'$ by $\Contour''$. With
respect to the direction of $K'$, the sub-polygon $\Contour'$ is
counter-clockwise oriented, while $\Contour''$ is clockwise
oriented. Further, $\Contour''$ contains the top point of
$\Contour$. We define that $\Contour'$ contains the sites
$s_{2},\dots,s_{k-1}$, while $\Contour''$ does not contain them.

Thus, the polyline~$K'$ partitions the instance
$(\Contour,\Sites,\Ports)$ into two sub-instances $\Instance'=(\Contour',\Sites',\Ports')$ and
$\Instance''=(\Contour'',\Sites'',\Ports'')$ such that
\begin{compactenum}[(1)]
  \item $\Sites'\cup \Sites'' = \Sites\setminus\{s_1,s_k\}$ and $\Ports' \cup \Ports'' = \Ports\setminus\{p_1,p_k\}$,
  \item the sites of $\Sites'$ lie in $\Contour'$ or on $K$ and the
    sites of $\Sites''$ lie in the interior of $\Contour''$,
  \item the ports of $\Ports'$ lie on the boundary of $\Contour'$ and
    the ports of $\Ports''$ lie on the boundary of $\Contour''$.
\end{compactenum}
Note that the sites $s_1$, $s_k$ and the ports $p_1$, $p_k$
  neither belong to $\Instance'$ nor to $\Instance''$, because they
  are already used by the fixed labels $\ell_1$ and $\ell_k$.  We
call $(\ell_1,\ell_k,K)$, which defines the polyline $K'$, the
\emph{separator} of $\Contour'$ and $\Contour''$. For two
sub-instances we say that they are \emph{disjoint} if the interiors of
their regions are disjoint.

In the following, we only consider sub-instances in which the convex
chain $K$ lies to the right of the line~$l$ through $s_1$ and $s_k$
pointing towards $s_k$ from $s_1$; we will show that these are sufficient for
decomposing any instance.
Put differently, the chain $K$ is a
convex part of the boundary of $\Contour'$ and a concave part of the
boundary of $\Contour''$. We call $\Instance'$ a \emph{convex}
sub-instance and $\Instance''$ a \emph{concave} sub-instance.

The line $l$ splits $\Contour''$ into three regions $A_1$, $A_2$ and
$A_3$; see Fig.~\ref{fig:basic-instances}(a). Let $A_2$ be the region
to the right of $l$ and let $A_1$ and $A_3$ be the regions to the left
of $l$ such that $A_1$ is adjacent to the leader of $\ell_1$ and $A_3$
is adjacent to the leader of $\ell_k$. Depending on the choice of
$\ell_1$ and $\ell_k$, the regions $A_1$ and $A_3$ may or may not exist. We
call $\Contour''$ the \emph{exterior} of $\Contour'$ and vice
versa.  We distinguish the following convex
instances.

\begin{compactenum}[(A)]
\item A convex instance has type A if there is a site $s\in A_1$ such
  that $\ell_1$ and the half-line~$h$ emanating from $s$ through $s_1$
  separates $K$ from the sites in $\Contour''$; see
  Fig~\ref{fig:basic-instances}(b).

\item A convex instance has type B if there is a site $s\in A_3$ such
  that $\ell_k$ and the half-line~$h$ emanating from $s$ through $s_k$
  separates $K$ from the sites in $\Contour''$; see
  Fig~\ref{fig:basic-instances}(c).
\end{compactenum}
For both types the chain $K$ is uniquely defined by the choice of
  $\ell_1$, $\ell_k$ and~$s$. Thus, type A and type B instances are
  uniquely defined by $\ell_1$, $\ell_k$ and $s$; we denote these
  instances by $\Instance_\TB[\ell_1,\ell_k,s]$ and
  $\Instance_\TC[\ell_1,\ell_k,s]$, respectively. We call $s$ the
\emph{support point} of the instance. In case that $\Contour''$ is
empty, the chain $K$ is already uniquely defined by $\ell_1$ and
$\ell_k$ and we write $\Instance_\TB[\ell_1,\ell_k,\bot]$ and
$\Instance_\TC[\ell_1,\ell_k,\bot]$.  Hence, we can enumerate all such
instances by enumerating all possible triples consisting of two labels
and one site. Since each label is defined by one port and one site, we
obtain $O(|\Sites|^3\dot|\Ports|^2)$ instances in total.
Note that instances of type B are symmetric to type A instances. 

For $k=2$ the chain consists of the sites $s_1$ and $s_2$ and
the support point is superfluous; such	 an instance is solely defined by
the labels $\ell_1$ and $\ell_2$ of $s_1$ and $s_2$, respectively.
We call these instances \emph{capstone instances} and denote them by
$\Instance_\IC[\ell_1,\ell_2]$; see
Fig.~\ref{fig:basic-instances}(d). 

We now show that any labeling can be composed into instances of these
three types.  We say that an instance is \emph{empty} if its set
$\Sites$ of sites is empty.  For a labeling $\mathcal L$ of an
instance $\Instance$ we write $\mathcal L\restriction_{\Instance'}$
for the labeling $\mathcal L'\subseteq \mathcal L$ that is restricted
to the sites and ports of the sub-instance~$\Instance'$ of
$\Instance$. Let $(\ell_1,\ell_k,K)$ denote the separator of
$\Instance'$. For a labeling~$\mathcal L'$ of a sub-instance
$\Instance'$ we require that any leader of any label~$\ell$ in
$\mathcal L'$ lies in the contour of the sub-instance not intersecting
the separator (the sites $s_2,\dots,s_{k-1}$ are excluded from this
restriction). In that case we say that~$\ell$ is contained in
$\Instance'$.
We further emphasize that the labels $\ell_1$ and $\ell_k$ are
contained in $\mathcal L'$, but they are fixed describing the contour
of~$\Instance'$ so that their sites and ports do not belong to the
sites and ports of $\Instance'$, respectively.  The next lemma states
how to decompose type A and type B instances into smaller type A, type
B and capstone instances. Fig.~\ref{fig:decomposition}(a) illustrates
Case~(\ref{complex-instances:typeA}).
Case~(\ref{complex-instances:typeB}) is symmetric to
Case~(\ref{complex-instances:typeA}).

\begin{figure}[t]
 \centering
 \includegraphics[page=2,width=\linewidth]{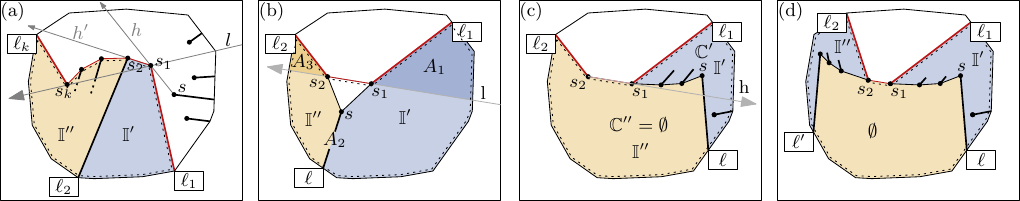}
 \caption{Decomposition of a convex instance~$\Instance$ (dashed
   polygon). (a)~Type A instance with $k>2$ is decomposed into
   capstone instance~$\Instance'$ and type A
   instance~$\Instance''$. (b)~Capstone instance is decomposed into
   two smaller capstone instances. (c)--(d)~Capstone instance is
   decomposed into type A and type B instances. }
   \label{fig:decomposition}
\end{figure}

\begin{lemma}\label{lem:complex-instances}
  Let $\Instance$ be a convex instance with separator
  $(\ell_1,\ell_k,K=(s_1,\dots,s_k))$ and $k>2$. Further, let
  $\mathcal L$ be a planar labeling of $\Instance$.

  \begin{compactenum}[(i)]
  \item\label{complex-instances:typeA} If $\Instance$ has type $A$, the label $\ell_2\in \mathcal L$
    of $s_2$ splits $\Instance$ into the disjoint sub-instances
    $\Instance'=\Instance_\IC[\ell_1,\ell_2]$ and
    $\Instance''=\Instance_\TB[\ell_2,\ell_k,s_1]$ such that any
    label~$\ell \in \mathcal L$ is contained in~$\Instance'$ or~$\Instance''$.
    \[c(\mathcal L) = c(\mathcal L\restriction_{\Instance_\IC[\ell_1,\ell_2]})+c(\mathcal L\restriction_{\Instance_\TB[\ell_2,\ell_k,s_1]})-c_1(\ell_2)\] 
  \item\label{complex-instances:typeB} If $\Instance$ has type $B$, the label $\ell_{k-1}\in \mathcal
    L$ of $s_{k-1}$ splits $\Instance$ into the two disjoint
    sub-instances $\Instance'=\Instance_\IC[\ell_{k-1},\ell_k]$ and
    $\Instance''=\Instance_\TC[\ell_1,\ell_{k-1},s_k]$ such that any
    label~$\ell \in \mathcal L$ is contained in $\Instance'$ or~$\Instance''$.
      \[c(\mathcal L) = c(\mathcal L\restriction_{\Instance_\IC[\ell_{k-1},\ell_{k}]})+c(\mathcal L\restriction_{\Instance_\TC[\ell_1,\ell_{k-1},s_k]})-c_1(\ell_{k-1})\] 
  \end{compactenum}
\end{lemma}

\begin{proof}
  We only argue for type A instances. Symmetric arguments hold for
  type B instances. Consider an arbitrary sub-instance
  $\Instance=(\Contour,\Sites,\Ports)$ of type A; see
  Fig.~\ref{fig:decomposition}(a).  Since $\ell_2$ connects two points
  of~$\Contour$'s boundary, it partitions $\Instance$ into two
  sub-instances~$\Instance'$ and~$\Instance''$ with labelings
  $\mathcal L\restriction_{\Instance'}$ and $\mathcal
  L\restriction_{\Instance''}$ such that any label of $\mathcal
  L\setminus\{\ell_2\}$ either is contained in~$\Instance'$
  or~$\Instance''$. Let $\Instance'$ be the instance containing $s_1$
  and $\Instance''$ the other one.  Obviously, $\Instance'$ forms the
  capstone instance $\Instance_\IC[\ell_1,\ell_2]$. We now show that
  $\Instance''$ forms
  instance~$\Instance''=\Instance_\TB[\ell_2,\ell_k,s_1]$ of type A.

  By definition of $\Instance$ the label~$\ell_1$ and the
  half-line~$h$ emanating from $s$ through $s_1$ separates the convex
  chain $K$ of $\Instance$ from the sites in the exterior of
  $\Instance$.  Because of the convexity of $K$, the half-line~$h'$
  emanating from $s_1$ through $s_2$ and the label $\ell_2$ separate
  the convex chain $K'=(s_2,\dots,s_k)$ from the sites in the exterior
  of $\Instance''$. Hence,
  $\Instance''=\Instance_\TB[\ell_2,\ell_k,s_1]$ has type A.

  The cost of $\mathcal L$ is composed of the costs of $\mathcal
  L\restriction_{\Instance'}$ and $\mathcal
  L\restriction_{\Instance''}$ as follows
  $c(\mathcal L) = c(\mathcal
  L\restriction_{\Instance_\IC[\ell_1,\ell_2]})+c(\mathcal
  L\restriction_{\Instance_\TB[\ell_2,\ell_k,s_1]})-c_1(\ell_2)
  $.
  To not count $\ell_2$ twice, we
  subtract~$c_1(\ell_2)$.
\end{proof}

The convex sub-instances $\Instance'$ and $\Instance''$ of the
previous lemma contain fewer sites than $\Instance$. Thus, recursively
applying Lemma~\ref{lem:complex-instances} decomposes $\Instance$ into
a set of type A and type B instances until all instances are capstone
instances. The next lemma states how to decompose capstone
  instances into smaller type A, type B and capstone
  instances. Fig.~\ref{fig:decomposition}(b) illustrates
  Case~(\ref{capstone:triangle}), Fig.~\ref{fig:decomposition}(c)
  illustrates Case~(\ref{capstone:typeA}), Case~(\ref{capstone:typeB})
  is symmetric to Case~(\ref{capstone:typeA}) and
  Fig.~\ref{fig:decomposition}(d) illustrates
  Case(\ref{capstone:typeAB}).

\begin{lemma}\label{lem:capstone}
  Let $\Instance=\Instance_\IC[\ell_1,\ell_2]$ be a capstone instance
  with separator $(\ell_1,\ell_2,K=(s_1,s_2))$. Further, let $\mathcal L$
  be a planar labeling of $\Instance$. One of the following statements
  applies for $\Instance$.

  \begin{compactenum}[(i)]
  \item \label{capstone:empty}The instance~$\Instance$ is empty. 
    \[c(\mathcal L)=c_1(\ell_1)+c_1(\ell_2)+c_2(\ell_1,\ell_2)\]
   \item \label{capstone:triangle}There is a label~$\ell$ in $\mathcal
    L$ such that any label in
    $\mathcal L$ is contained in one of the two disjoint capstone instances
    $\Instance_\IC[\ell_1,\ell]$ and $\Instance_\IC[\ell,\ell_2]$.
    \[c(\mathcal L) = c(\mathcal L\restriction_{\Instance_\IC[\ell_{1},\ell]})+c(\mathcal L\restriction_{\Instance_\IC[\ell,\ell_2]})-c_1(\ell)\] 
  \item \label{capstone:typeA}There is a label~$\ell\in \mathcal
    L$ s.t.\ any label in
    $\mathcal L$ is contained in
    $\Instance_\TB[\ell_1,\ell,s_2]$.
   \[
     c(\mathcal L) = c(\mathcal L\restriction_{\Instance_\TB[\ell_1,\ell,s_2]})+c_2(\ell,\ell_2)
   \]
  \item \label{capstone:typeB}There is a label~$\ell \in \mathcal
    L$ s.t.\ any label in
    $\mathcal L$ is contained in
    $\Instance_\TC[\ell,\ell_2,s_1]$.
     \[
     c(\mathcal L) = c(\mathcal L\restriction_{\Instance_\TC[\ell,\ell_2,s_1]})+c_2(\ell_1,\ell)
   \] 
 \item \label{capstone:typeAB}There are labels~$\ell,\ell'\in \mathcal
   L$ with $\ell\neq \ell'$ s.t.\ any label in $\mathcal L$ is contained in
   either $\Instance_\TB[\ell_1,\ell,s_2]$ or
   $\Instance_\TC[\ell',\ell_2,s_1]$.
     \[
     c(\mathcal L) = c(\mathcal
     L\restriction_{\Instance_\TB[\ell_1,\ell,s_2]})+c(\mathcal
     L\restriction_{\Instance_\TC[\ell',\ell_2,s_1]})+c_2(\ell,\ell')
   \]  
  \end{compactenum}

\end{lemma}

\begin{proof}
  If $\Instance=(\Contour,\Sites,\Ports)$ is empty, the cost of
  $\mathcal L$ are composed by $c(\mathcal
  L)=c_1(\ell_1)+c_2(\ell_2)+c_2(\ell_1,\ell_2)$, which yields~Case
  (\ref{capstone:empty}).

  So assume that $\Instance$ is not empty. The line $l$ through~$s_1$
  and $s_2$ splits $\Contour$ into three regions $A_1$, $A_2$
  and~$A_3$; see Fig.~\ref{fig:decomposition}(b). Let~$A_2$ be the
  region to the left of $l$ (going along $l$) and let $A_1$ and~$A_3$
  be the regions to the right of $l$ such that $A_1$ and $A_3$ are
  adjacent to the leaders of~$\ell_1$ and~$\ell_2$,
  respectively. Further, let $p_1$ and $p_2$ be the ports of~$\ell_1$
  and~$\ell_2$, respectively.
  
  First assume that there is a site~$s$ with label $\ell\in\mathcal L$
  such that the \emph{separating triangle} $\Delta(s_1,s,s_2)$ lies in
  $\Contour$ and its interior is not intersected by any label in
  $\mathcal L$; the triangle must lie in $A_2$. Together with $\ell$,
  it partitions~$\Instance$ into the two sub-instances $\Instance'$
  and $\Instance''$. They form the capstone instances
  $\Instance'=\Instance_\IC[\ell_1,\ell]$ and
  $\Instance''=\Instance_\IC[\ell,\ell_2]$ as in Case~(\ref{capstone:triangle}). It is easy to see that $c(\mathcal L) =
  c(\mathcal L\restriction_{\Instance_\IC[\ell_{1},\ell]})+c(\mathcal
  L\restriction_{\Instance_\IC[\ell,\ell_2]})-c_1(\ell).$ We subtract
  $c_1(\ell)$ to not count $\ell$ twice.

  So assume that there is no such separating triangle, which implies
  that $A_1 \cup A_3$ contains sites. If this were not the case, $A_2$
  would contain a site, because $\Instance$ is not empty. Among these
  sites we easily could find a site forming a separating triangle. In
  particular due to Assumption~1 such a triangle cannot be
  intersected by any text box.  We distinguish three cases how~$A_1$
  and $A_3$ contain sites.

  \textit{Case: $A_1$ contains a site, but $A_3$ is empty.} Let $s$ be a site in
  $A_1$. We denote its label by $\ell$ and the port of $\ell$ by
  $p$. We choose $s$ and $\ell$ such that no label of any site in $A_1$
  succeeds $\ell$ in the radial ordering. Let $K$ be the convex chain
  of sites in $A_1$ such that $K$ starts at $s_1$, ends at $s$, and
  any site in $A_1$ lies in the counterclockwise oriented
  polygon~$\Contour'$ defined by $p_1$, $K$, $p$ and the part of
  $\Contour$ in between $p$ and $p_1$; see
  Fig.~\ref{fig:decomposition}(c). All sites in $A_2$ also lie in
  $\Contour'$, because otherwise we could find a separating
  triangle. Hence, the clockwise oriented polygon~$\Contour''$
  defined by $p_2$, $s_2$, $s_1$, $K$, $p$ and the part of $\Contour$
  in between $p$ and $s_2$ does not contain any site. By the choice of
  $\ell$ all labels must be contained in the instance induced by
  $\Contour'$. In particular the half-line~$h$ emanating from $s_2$
  through $s_1$ separates $K$ from the sites in the exterior of
  $\Instance'$. Hence $\Instance'$ is the type A instance
  $\Instance_\TB[\ell_1,\ell,s_2]$. This yields Case~(\ref{capstone:typeA}) and 
  $
    c(\mathcal L) = c(\mathcal L\restriction_{\Instance_\TB[\ell_1,\ell,s_2]})+c_2(\ell,\ell_2).
  $

  \textit{Case: $A_3$ contains a site, but $A_2$ is empty.} Symmetrically, we
  construct a type B instance
  $\Instance'=\Instance_\TC[\ell,\ell_2,s_1]$, which yields
  Case~(\ref{capstone:typeB}).

    \textit{Case: Both $A_1$ and $A_2$ contain sites.} Combining the
    construction of the previous two cases, we obtain two labels
    $\ell$ and $\ell'$ that define the type A and type B instances
    $\Instance'=\Instance_\TB[\ell_1,\ell,s_2]$ and $\Instance''=\Instance_\TC[\ell',\ell_2,s_1]$ of
    case (\ref{capstone:typeAB}); see Fig.~\ref{fig:decomposition}(d).
\end{proof}

Lemma~\ref{lem:complex-instances} and Lemma~\ref{lem:capstone}
describe how to decompose an arbitrary convex sub-instance $\Instance$
into a set of empty capstone
instances. The next lemma implies that any labeling of any instance
$\Instance$ of \ContourLabeling can be decomposed into empty capstone
instances. 

\begin{lemma}\label{lem:base-case}
  Let $\Instance$ be an instance of
  \ContourLabeling and let $\mathcal L$ be a planar labeling of
  $\Instance$. The first leader $\ell$ and the last leader $\ell'$ in
  the radial ordering of $\mathcal L$ define a type A instance
  $\Instance'=\Instance_\TB[\ell,\ell',\bot]$ such that the exterior
  of $\Instance'$ is empty and
   $
   c(\mathcal L) = c(\mathcal
   L\restriction_{\Instance'})+c_2(\ell,\ell').
   $  
\end{lemma}

\begin{figure}
 \begin{minipage}[b]{0.49\linewidth}
   \centering
   \includegraphics[page=3]{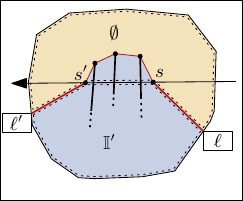}
   \caption{Illustration of Lemma~\ref{lem:base-case}}
   \label{fig:base-case}
 \end{minipage}
 \hfill
 \begin{minipage}[b]{0.49\linewidth}
   \centering
   \includegraphics[page=5]{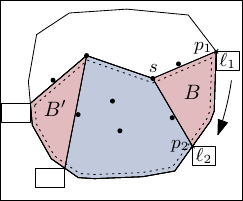}
   \caption{Illustration of Bundles.}
   \label{fig:bundles}
 \end{minipage}
\end{figure}

\begin{proof}
  Let~$\Instance =(\Contour,\Sites,\Ports)$. Further, let~$s$ and $s'$ be the sites and let $p$ and $p'$ be the ports of the two labels
  $\ell$ and $\ell'$, respectively. The polyline $(p,s,s',p')$
  partitions~$\Contour$ into two sub-polygons; see
  Fig.~\ref{fig:base-case}. Let $\Contour'$ be the
  counterclockwise oriented polygon and $\Contour''$ be the clockwise
  oriented polygon. Let $\Sites''$ be the sites in $\Contour''$. The
  port of any label in $\mathcal L$ lies on the boundary
  of~$\Contour'$, because otherwise $\ell$ and $\ell'$ would not be
  the first and last labels in the radial ordering of $\mathcal L$,
  respectively. Hence, any leader of any label with site in $\Sites''$
  must intersect the line segment $\overline{ss'}$. This in particular
  implies that any of these sites must lie to the right of the line
  $l$ that goes through $s$ and $s'$ in that direction. Let
  $H$ be the convex hull of $\Sites''\cup\{s,s'\}$. Removing
  $\overline{ss'}$ from $H$, we obtain the desired convex chain $K$, which
  lies to the right of~$l$. Together with $\ell$ and $\ell'$ it forms
  the type A instance $\Instance_\TB[\ell,\ell',\bot]$.
\end{proof}

\subsection{Dynamic Programming}\label{sec:dyn-prog}
Applying the results of the previous section we present a dynamic
programming approach that solves \ContourLabeling with fixed ports
optimally.  For type A, type B and capstone instances the approach
creates the three tables $T_\TB$, $T_\TC$ and $T_\IC$ storing the
optimal costs of the considered instances, respectively.  We call an
instance \emph{valid} if the two labels $\ell_1$ and $\ell_k$ defining
the separator do not intersect and comply with~T\ref{crit:shadowing}.

\textsc{Step~1.}~We compute all valid instances of type~A and type~B,
and all valid capstone instances.

\textsc{Step~2.}~We compute the optimal costs for all convex
sub-instances.  Let $\Instance$ be the currently considered instance
of size $i\geq 0$ with separator $(\ell_1,\ell_k,K=(s_1,\cdots,s_k))$,
where the \emph{size} of $\Instance$ is the number of sites contained in~$\Instance$; recall that $s_1$ and $s_k$ do not belong to~$\Instance$.
Considering the instances in non-decreasing order of their sizes, we
can assume that we have already computed the optimal costs for all
convex instances with size less than $i$. We distinguish the two main
cases $k=2$ and $k>2$.

\textbf{Case} $\mathbf{k=2.}$ The instance~$\Instance$ forms the capstone instance
$\Instance_\IC[\ell_1,\ell_2]$. Let~$s_1$ and $s_2$ be the sites
of~$\ell_1$ and $\ell_2$, respectively. Following
Lemma~\ref{lem:capstone} we apply five steps.

\begin{compactenum}[(1)]
\item If $\Instance_\IC[\ell_1,\ell_2]$ is not empty, we set $w_1:=\infty$ and otherwise 
  \[w_1:=c_1(\ell_1)+c_1(\ell_2)+c_2(\ell_1,\ell_2).\]
\item We consider every site~$s$ in $\Instance_\IC[\ell_1,\ell_2]$ such that
  the \emph{separating triangle}~$\Delta(s_1,s_2,s)$ lies in the
  region of~$\Instance$ and does not contain any other site. For all
  such sites we determine each label candidate $\ell$ that partitions
  $\Instance$ into valid disjoint capstone instances
  $\Instance_\IC[\ell_1,\ell]$ and $\Instance_\IC[\ell,\ell_2]$. Let
  $D$ denote the set of all those labels.
  \[w_2 := \min_{\ell\in D}
  T_\IC[\ell_1,\ell]+T_\IC[\ell,\ell_2]-c_1(\ell).\]
\item We determine every label $\ell$ of $\Instance_\IC[\ell_1,\ell_2]$ such
  that every site of $\Instance$ lies in
  $\Instance_\TB[\ell_1,\ell,s_2]$. Let $D$ denote the set of those
  labels.  
  \[w_3 := \min_{\ell\in D} T_\TB[\ell_1,\ell,s_2]+c_2(\ell,\ell_2).\]
\item We determine every label $\ell$ of $\Instance_\IC[\ell_1,\ell_2]$ such
  that every site of $\Instance$ lies in
  $\Instance_\TC[\ell,\ell_2,s_1]$. Let $D$ denote the set of those
  labels. 
  \[w_4 := \min_{\ell\in D} T_\TC[\ell,\ell_2,s_1]+c_2(\ell_1,\ell).\]
\item We determine every pair $(\ell,\ell')$ of intersection-free labels
  of $\Instance_\IC[\ell_1,\ell_2]$ such that every site of
  $\Instance[\ell_1,\ell_2]$ lies in $\Instance_\TB[\ell_1,\ell,s_2]$
  or $\Instance_\TC[\ell',\ell_2,s_1]$. Let $D$ denote the set of
  those pairs labels.  \[w_5 := \min_{(\ell,\ell')\in D}
  T_\TB[\ell_1,\ell,s_2]+T_\TC[\ell',\ell_2,s_1]+c_2(\ell,\ell').\]
\end{compactenum}
In any of the above cases we choose the labels of the candidate set
$D$ such that the considered instances are valid. Further, if $D$ is
empty in sub-step ($i$), we set $w_i:=\infty$ ($2\leq i \leq 5$).  We set
$T_\IC[\ell_1,\ell_2] := \min_{1\leq i\leq 5}\{w_i\}$.

\textbf{Case} $\mathbf{k>2.}$ Following Lemma~\ref{lem:complex-instances} we distinguish two
sub-cases. 
If $\Instance$ is a type $A$ instance $\Instance_\TB[\ell_1,\ell_k,s]$ (where possibly $s=\bot$), we determine every possible
label $\ell$ for site $s_2$.  Let $D$ be the set of those labels. If $D$ is empty, we set 
$T_\TB[\ell_1,\ell_k,s]:=\infty$ and otherwise
\[T_\TB[\ell_1,\ell_k,s] := \min_{\ell \in D} T_\IC[\ell_1,\ell] +
T_\TB[\ell,\ell_k,s_1]-c_1(\ell).\] If $\Instance$ has type B,
we analogously define $D$ for $s_{k-1}$. If $D$ is empty, we set $T_\TC[\ell_1,\ell_k,s]:=\infty$ and otherwise
\[T_\TC[\ell_1,\ell_k,s] := \min_{\ell \in D} T_\IC[\ell,\ell_{k}] +
T_\TC[\ell_1,\ell,s_{k}]-c_1(\ell).\]
In any of the cases above we call the elements in~$D$ the
\emph{descendants} of the given instance.

\textsc{Step 3.} After computing the optimal costs of all convex
instances, we enumerate all possible choices $(\ell,\ell')$ of first
and last labels in the radial ordering. Let $D$ denote the set of
those choices.  By Lemma~\ref{lem:base-case} any choice $(\ell,\ell')
\in D$ forms a convex type A instance
$\Instance'=\Instance_\TB[\ell,\ell',\bot]$.  Hence, the optimal
cost~$\OPT(\Instance)$ of $\Instance$ are
\[
\OPT(\Instance) = \min_{(\ell,\ell') \in D}
T_\TB[\ell,\ell',\bot]+c_2(\ell,\ell').
\]
Recall that if the convex chain $K$ of $\Instance'$ has length 2, then
we have $\Instance'=\Instance_\IC[\ell,\ell']$ and therefore
$T_\TB[\ell,\ell',\bot]=T_\IC[\ell,\ell']$.  If $\OPT(\Instance) =
\infty $, we return that $\Instance$ does not admit a contour
labeling.

\textsc{Step 4.} Also storing the according descendants for each instance, we apply a standard
backtracking approach for dynamic programming to obtain the corresponding 
labeling~$\mathcal L$.

\begin{theorem}
  For an instance
  $\Instance=(\Contour,\Sites,\Ports)$
  contour labeling, the problem \ContourLabeling can be solved in
  $O(|\Sites|^4|\Ports|^4)$ time.
\end{theorem}

\begin{proof}
  We first show the correctness and then argue 
  the running time. Let $\Instance$ be an arbitrary instance
  of \ContourLabeling.

  \textit{Correctness.}  Let~$\mathcal L$ be a labeling constructed by
  our algorithm. We first show that $\mathcal L$ is planar proving
  the conditions of Lemma~\ref{lem:shadow-free}.

  By construction we ensure that for a convex instance $\Instance$ the
  labels of $\Instance$'s separator comply with
  T5. In particular this implies that any two
  consecutive labels in $\mathcal L$ comply with
  T5, which implies that $\mathcal L$ is a
  staircase labeling.  When computing the descendants $D$ for an instance,
  we ensure that $D$ does not contain any label that
  intersects the separator of that instance.  By induction this
  implies that no leader in $\mathcal L$ intersects any other leader
  in~$\mathcal L$. Further, we ensure that no two
  consecutive labels intersect.  Hence, by Lemma~\ref{lem:shadow-free}
  the labeling~$\mathcal L$ is planar.

  Finally, we prove that $\mathcal L$ is an optimal labeling of
  $\Instance$. By
  Lemma~\ref{lem:complex-instances},~\ref{lem:capstone}
  and~\ref{lem:base-case} any labeling can be recursively decomposed
  into type A, type~B, and capstone instances. Our algorithm enumerates
  all these instances. For each such instance~$\Instance'$ we search
  for the label(s) $\ell$ ($\ell'$) that split $\Instance'$ into
  smaller type A, type B and capstone instances. Since we enumerate
  all such labels, we also find the label minimizing the cost
  of $\Instance'$. 
 
  \textit{Running Time.} We now analyze the running time step by step.
  
  \textsc{Step 1.} We create all possible convex instances of type A and
  type~B, and all capstone instances by (conceptually) enumerating all
  tuples~$(\ell_1,\ell_2)$ and~$(\ell_1,\ell_2,s)$, where~$\ell_1$ and
  $\ell_2$ are labels based on the ports and sites of $\Instance$ and
  $s\in S$. Since each label is defined by a site and a port, we
  enumerate $O(|\Sites|^3\cdot |\Ports|^2)$ tuples. For each instance
  we also compute the convex chain $K$ and the sites contained in the
  instance, which needs $O(|\Sites|)$ time assuming that the
    sites are sorted by their $x$-coordinates.  Sorting the
  instances by size needs $O(|\Sites|^3\cdot |\Ports|^2 (\log |\Sites|
  + \log|\Ports|))$ time.
 
  \textsc{Step 2.} Handling a single capstone instance we need
  $O(|\Sites|^2 |\Ports|^2)$ time: For the cases (1)--(4) there are
  $O(|\Sites|\cdot |\Ports|)$ descendants, while for case (5) there
  are $O(|\Sites|^2 \cdot |\Ports|^2)$ descendants. Since we consider
  $O(|\Sites|^2\cdot|\Ports|^2)$ many instances, we obtain
  $O(|\Sites|^4|\Ports|^4)$ running time in total for capstone
  instances.

  Handling a single type A or type B instance there are $O(|\Ports|)$
  descendants, because the site of the descendants is fixed. Since we
  consider $O(|\Sites|^3\cdot |\Ports|^2)$ such instances, we obtain
  $O(|\Sites|^3|\Ports|^3)$ running time, which is dominated by handling the capstone instances.  %

  \textsc{Step 3.} Enumerating all pairs of first and last labels in
  the radial ordering can be done in $O(|\Sites|^2|\Ports|^2)$~time.

  \textsc{Step 4.} Storing for each instance its descendant of lowest
  cost, we can do backtracking in linear~time.
  
  \noindent Altogether, Step 2 dominates  with
  $O(|\Sites|^4|\Ports|^4)$ running time.
\end{proof}

The criteria mentioned in Section~\ref{sec:criteria} can be easily
patched into the approach. If a criterion should become a hard
constraint not to be violated, we simply exclude any sub-instance
violating this specific criterion. For example, to enforce
monotonicity (L\ref{crit:monotone}) we remove any sub-instance whose defining labels violate
this criterion. Similarly, if the criterion should become a
soft-constraint, we do not exclude the sub-instance, but include its compliance with this criterion into its
cost---provided that it can be modeled by the cost functions $c_1$ or $c_2$. This is true for all criteria listed in Section~\ref{sec:criteria}.  

\section{Algorithm Engineering}\label{sec:engineering}
Initial experiments showed that a naive implementation of the dynamic
programming approach does not yield reasonable running times, which matches the high asymptotic running times.
In this section, we therefore describe how the approach can be
implemented efficiently to prevent these problems for instances of
realistic input sizes of $|S| \le 70$.

\subsection{Bundling}
Instead of considering each possible label individually, we \emph{bundle}
labels and redefine the different types of instances based on bundles
instead of single labels.  More precisely, consider two labels
$\ell_1$ and $\ell_2$ of the same site $s$ such that $\ell_1$ precedes
$\ell_2$ in the radial ordering; see
Fig.~\ref{fig:bundles}. Let $p_1$ and $p_2$ be the ports of
$\ell_1$ and $\ell_2$, respectively. Further, let~$R$ be the region
enclosed by $\ell_1$, $\ell_2$ and the part of the
contour~$\Contour$ that lies between $p_1$ and $p_2$ (going in
clockwise direction around $\Contour$).  We call $\ell_1$ and
$\ell_2$ \emph{equivalent} if the region~$R$ does not contain any
site. Hence, $\ell_1$ can be \emph{slid} along $\Contour$ to $\ell_2$
without passing any other site.  A \emph{bundle} of $s$ is a set~$B$
of labels of~$s$ that are pairwise equivalent.  A set $B$ is
\emph{maximal} if there is no bundle $B'$ with $B\subsetneq B'$. We order
the labels according their radial ordering. To ease the description we assume without loss of
generality that there is no bundle that contains the top point
of~$\Contour$; we can \emph{split} bundles at the top point
of~$\Contour$, if necessary.

As presented in Section~\ref{sec:struc-properties} we describe a
sub-instance $\Instance_\TB[\ell,\ell',t]$
($\Instance_\TC[\ell,\ell',t]$, $\Instance_\IC[\ell,\ell']$) by two
labels $\ell$ and $\ell'$ and by a support point~$t$. We now
generalize this concept to bundles. For two bundles $B$ and $B'$ of
two sites $s$ and $s'$ the \emph{instance set}
$\InstanceSet_\TB[B,B',t]$ contains every instance
$\Instance_\TB[\ell,\ell',t]$ with $\ell\in B$ and $\ell'\in B'$. We
analogously define the sets $\InstanceSet_\TC$ and $\InstanceSet_\IC$.
Using these instance sets, we speed up the steps of our algorithm
without losing optimality as follows.

\textbf{\textsc{Step 1}.}  Instead of creating every instance
of every type, we create all possible instance sets $\InstanceSet_\TB$,
$\InstanceSet_\TC$ and $\InstanceSet_\IC$ based on the maximal bundles
of the input instance. We note that we do not compute the instance
sets explicitly, but each such set~$\InstanceSet$ is uniquely
described by the two labels that occur first and last in
$\InstanceSet$ with respect to the radial ordering.
Since any two instances of an instance set
contain the same sites, we only compute these contained sites once per
instance set. We further exclude any instance set that does not permit
a labeling at all. To that end, we test whether there is enough space
along the enclosed contour for all labels.

\textbf{\textsc{Step 2} \& \textsc{Step 3}.} Instead of computing the
optimal cost for each individual instance, we iteratively compute lower
and upper bounds of the optimal cost of the instance sets and refine
these until we obtain the optimal cost.

More precisely, in each iteration we compute for each instance
set~$\InstanceSet$ a lower bound $L$ and an upper bound~$U$ for the
optimal costs of the instances in $\InstanceSet$, i.e., for each
instance $\Instance\in \InstanceSet$ it holds $L \leq \OPT(\Instance)
\leq U$. To that end we define for two bundles $B$ and $B'$ the cost
functions $c_1$ and $c_2$ as follows.
\begin{align*}
  c_1(B)=&(\min_{\ell \in B}c_1(\ell),\max_{\ell \in B}c_1(\ell))\\
  c_2(B,B')=&(\min_{\ell \in B,\ell'\in B'}c_2(\ell,\ell'),\max_{\ell
    \in B,\ell'\in B'}c_2(\ell,\ell'))
\end{align*}
Interpreting the result of $c_1(B)$ and $c_2(B,B')$ as two-dimensional
vectors we compute the tables $T_\TB$, $T_\TC$ and $T_\IC$ in the same
manner as before, but this time we use instance sets instead of
instances and bundles instead of single labels.  We analogously adapt \textsc{Step 3}.

Hence, we obtain for each instance set lower and upper bounds. In
particular we obtain a lower bound $\bar L$ and an upper bound $\bar
U$ for the given input instance. In case that $\bar L=\infty$, the
input instance does not admit a labeling and we can abort the
algorithm. Otherwise, we collect all instance sets of \textsc{Step 3}
whose lower bound does not exceed $\bar U$ and start a standard
backtracking procedure to collect all instance sets whose lower bound
does not exceed $\bar U$. Any other instance set is omitted, because
it cannot be part of the optimal solution. Afterwards, we split each
bundle $B$ into two half-sized bundles $B_1$ and $B_2$, i.e., let
$\ell_1,\dots,\ell_h$ be the labels in $B$ with respect to the radial
ordering. We then set $B_1=\{\ell_1,\dots,\ell_m\}$ and
$B_2=\{\ell_{m+1},\dots,\ell_{h}\}$ where $m=\lfloor
\frac{h}{2}\rfloor$.  Thus, based on those new bundles we obtain new
corresponding instance sets. We start the next iteration with the
newly created instance sets.

We repeat this procedure until each bundle contains exactly
one label. Thus, the bounds $\bar L$ and $\bar U$ equal the
optimal cost of the input instance. Doing backtracking we construct
the according labeling.

\begin{figure}[t]
  \centering
  \includegraphics[page=8]{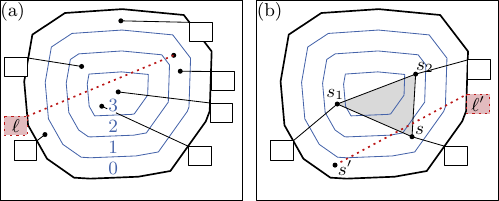}
  \caption{Shells (blue). (a) The label $\ell$ does not directly
    strive outwards and disturbs the overall appearance of the
    drawing. (b)~No label candidate~$\ell'$ of $s'$ can intersect the
    separating triangle~$\Delta(s_1,s_2,s)$ without intersecting a
    shell of higher level.}
  \label{fig:bundles-shells}
\end{figure}

\subsection{Shells}
We now describe two adaptations that further accelerate
the approach. In contrast to the previous section these techniques do
not necessarily preserve optimality. In
Section~\ref{sec:experiments} we show that in practice our heuristics
achieve near-optimal solutions in reasonable time.

 We observe that
leaders typically point outwards without passing through the
\emph{center} of the figure; see Fig.~\ref{fig:bundles-shells}(a).  This guarantees short leader lengths
and leaders fit in the overall appearance of the figure.  We use this
as follows. Based on the contour~$\Contour$ we construct a set of
offset polygons in the interior of~$\Contour$ such that they have a
pre-defined distance to each other; see
Fig.~\ref{fig:bundles-shells}(a). We call these polygons
\emph{shells}. The level of a shell is the number of shells containing
that shell and the level of a site is the level of the
innermost shell containing~$s$.

We require for every site that its leader does not intersect a shell
with higher level. Hence, we can exclude any label that violates this
requirement. The more shells are used the more labels are excluded
improving the running time. However, it also becomes more likely that
the optimal labeling gets lost. In our experiments we have chosen the
shells such that less than $0.9\%$ of the labels in handmade drawings
violate this property; see also Section~\ref{sec:experiments}.

We further speed up Step 2(2). To that end
let~$\Instance_\IC[\ell_1,\ell_2]$ be the currently considered
capstone instance and let~$s_1$ and $s_2$ be the sites of~$\ell_1$ and
$\ell_2$, respectively. Further, let $D$ be the set of descendants. We
only consider descendants that have a level that is at least as high
as the level of~$s_1$ and $s_2$. If such descendants do not exist,
  we take those with highest level.  In case that the shells are
convex and nested this particular adaption does not have any impact on
the achievable cost, because for any such site~$s$ the
triangle~$\Delta(s_1,s_2,s)$ is contained in the shell of~$s_1$, $s_2$
or $s$. It therefore cannot be intersected by any leader of a
descendant with a site that has a lower level; see Fig.~\ref{fig:bundles-shells}(b).

\newcommand{\scaleExamples}{0.3}

\begin{figure}[t]
  \centering
  \begin{tabular}{cc}
  \includegraphics[scale=\scaleExamples]{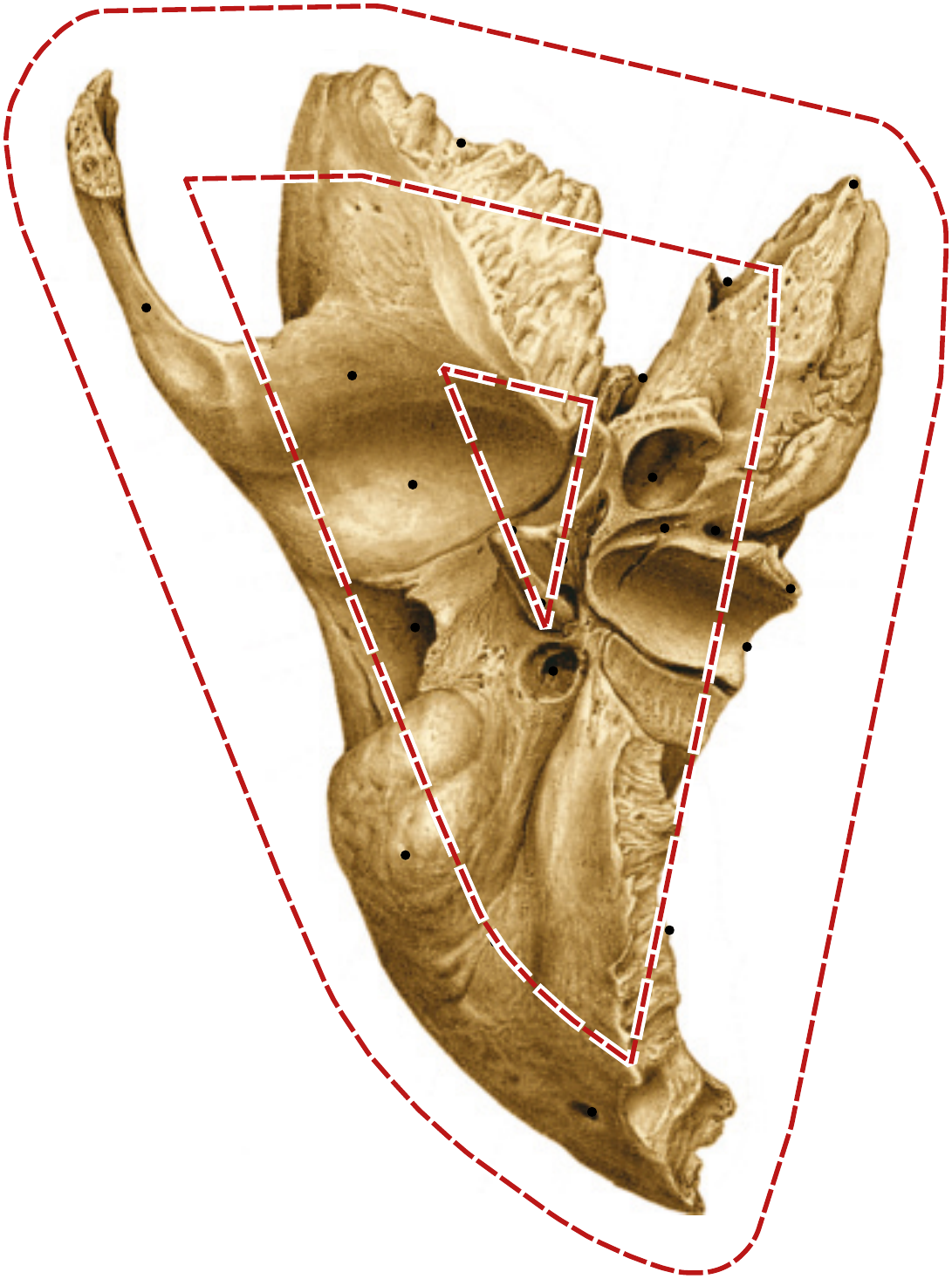}&    
 \includegraphics[scale=\scaleExamples]{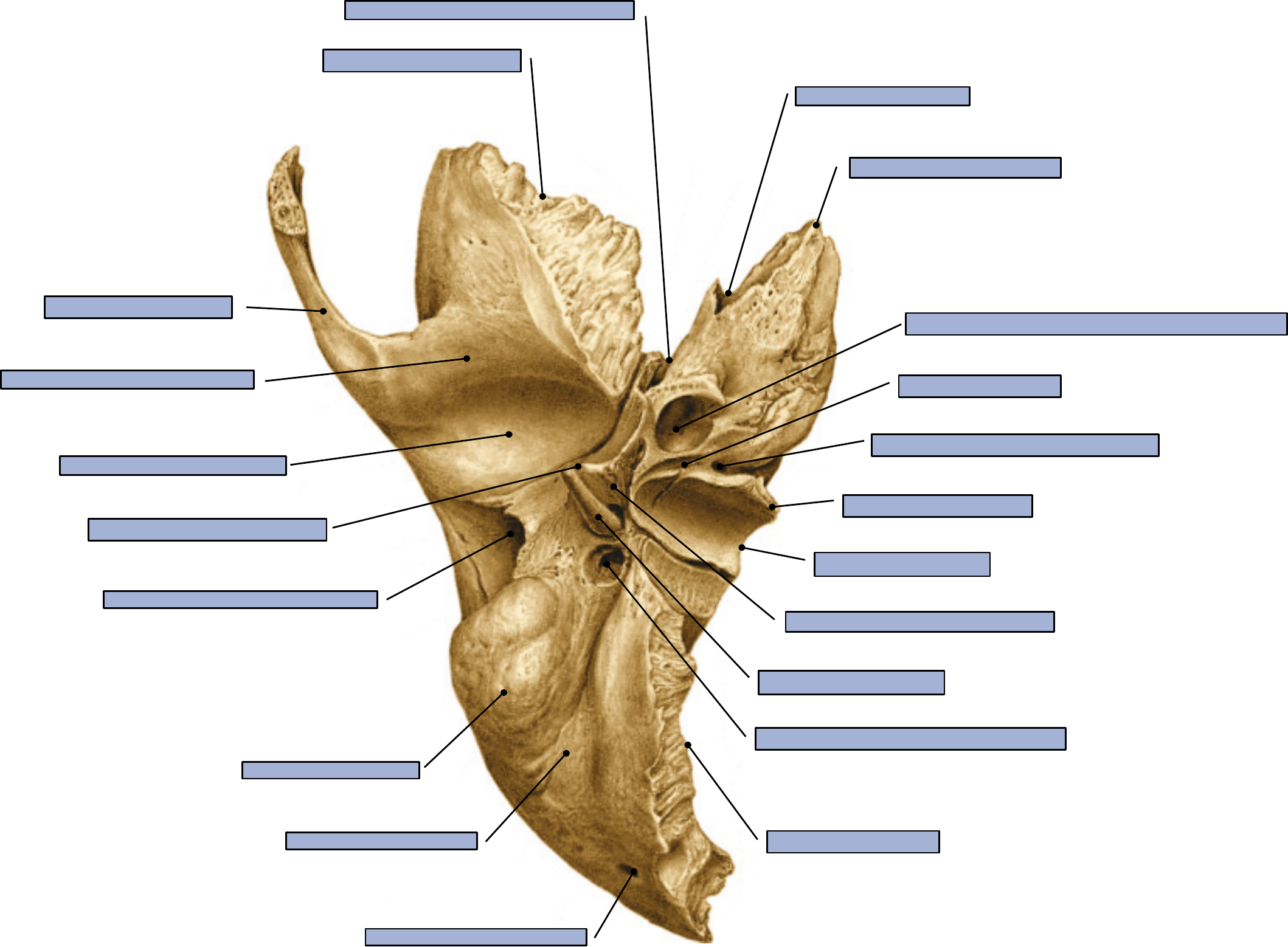}\\
  Contour \& Shells & Original \\ \\
 \includegraphics[scale=\scaleExamples]{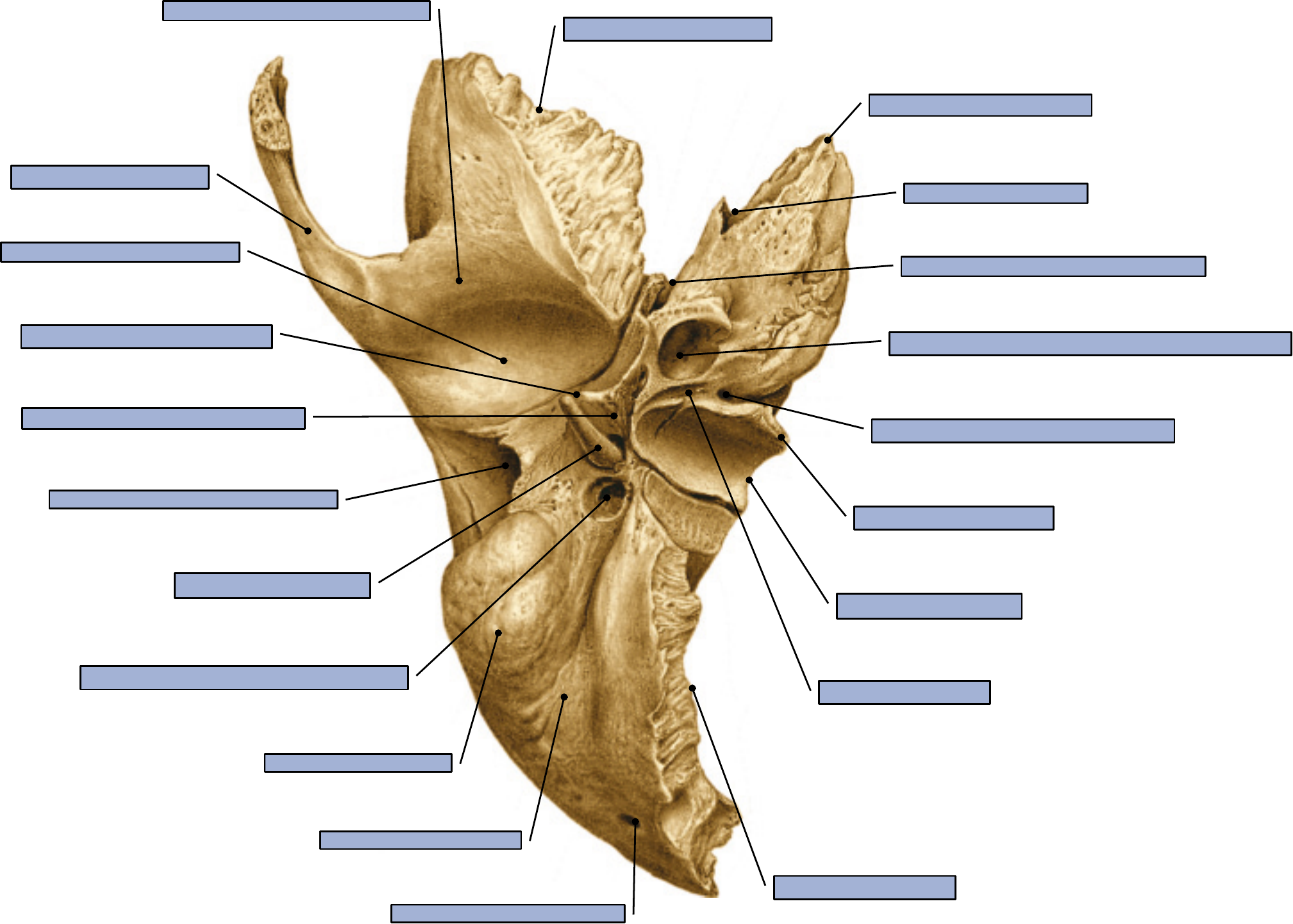}&
 \includegraphics[scale=\scaleExamples]{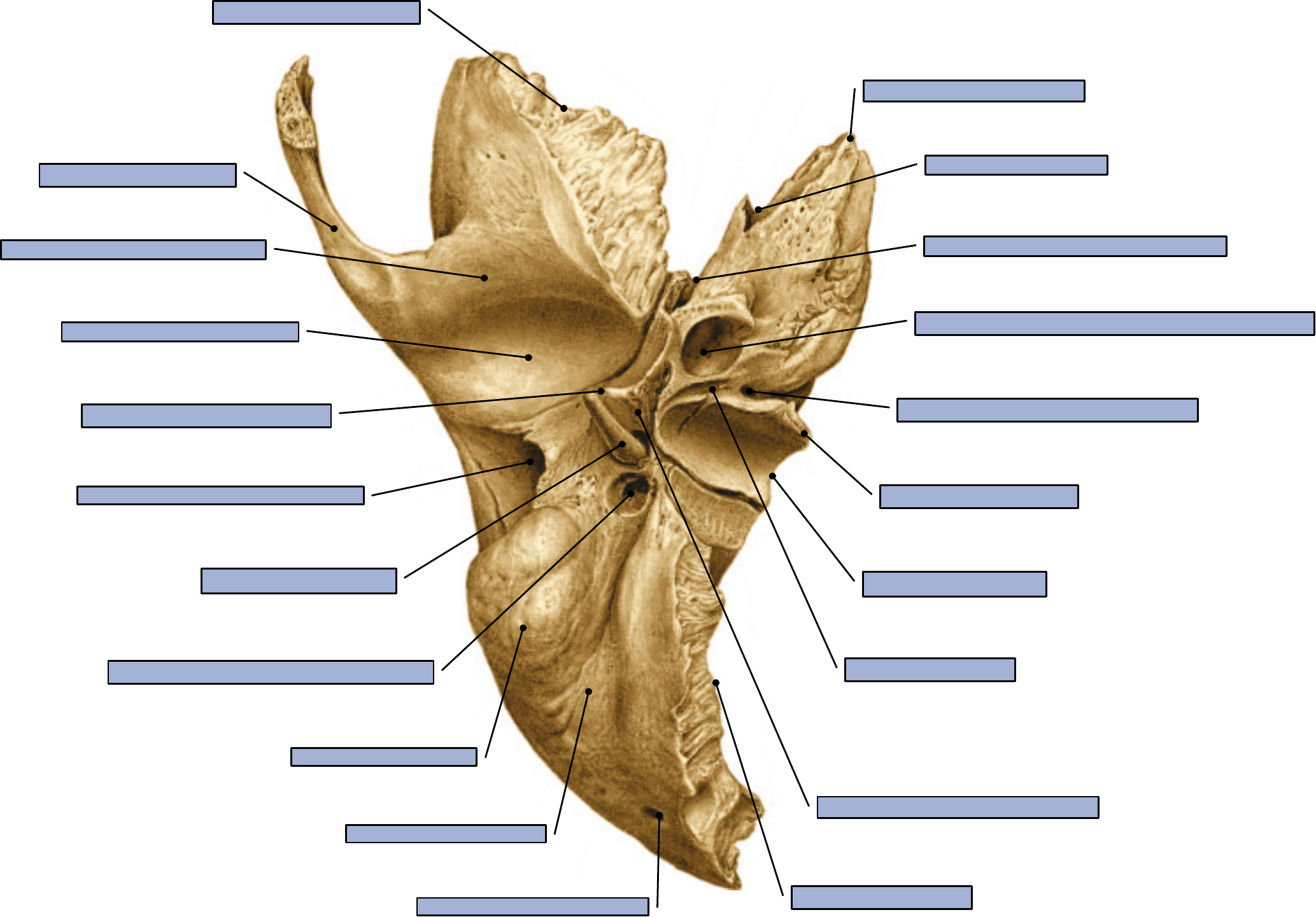}\\
  \OPT (10 sec.),  \CH (7 sec.) & \SCH (2 sec.), \TSCH (1 sec.) 
 \end{tabular}
 \caption{%
    \OPT and \CH as well as \SCH and \TSCH produced
   the same labelings, respectively. Source: Paulsen, Waschke, Sobotta
   Atlas Anatomie des Menschen, 23.Auflage 2010 \copyright Elsevier
   GmbH, Urban \& Fischer, München.}
  \label{fix:example:approaches}
\end{figure}

\subsection{Miscellaneous}
We can further speed-up the approach as follows.

\textsc{Simple-Instances.} Initial experiments showed that
non-capstone instances are more of theoretical interests proving the
optimality of the approach, but typically the optimal labeling can be
decomposed into capstone instances. Hence, it lends itself to only
consider capstone instances; particularly Step 2(3)--(5) are
omitted. The asymptotic running time remains the same,
  because handling capstone instances dominates the running time.

\textsc{One-Sided-Instances.} Assuming criterion~G\ref{crit:sides} we
can apply the following speed-up technique preserving the optimality
of our approach. Consider a capstone instance
$\Instance_\IC[\ell_1,\ell_2]$ such that both labels $\ell_1$ and
$\ell_2$ are right labels. Let~$D$ be the descendants computed in Step
2(2). Indeed we only need to consider the descendants in $D$ with the
leftmost site~$s$ among all those descendants. It preserves
optimality, because no descendant in $D$ of any other site can
intersect the separating triangle $\Delta(s_1,s_2,s)$; due to
criterion~G\ref{crit:sides} they all lie to the right of the vertical
line through $s$. Here $s_1$ and $s_2$ are the sites of $\ell_1$ and
$\ell_2$. Symmetrically, we can apply the same technique for left
labels $\ell_1$ and $\ell_2$ and the descendants with the rightmost
site among all descendants in~$D$.  The asymptotic running time
  remains the same, because handling capstone instances with left and
  right labels dominates the running time.

\textsc{Small-Triangles.} Computing the set~$D$ of descendants of a
capstone instance $\Instance_\IC[\ell_1,\ell_2]$ in Step 2(2), any
site $s$ in $\Instance_\IC[\ell_1,\ell_2]$ is considered that forms an
empty separating triangle $\Delta(s_1,s_2,s)$. Hereby $s_1$ and $s_2$
are the sites of $\ell_1$ and $\ell_2$. For this speed-up technique we
only consider the site~$s$ with smallest triangle~$\Delta(s_1,s_2,s)$
among those sites reducing the number of descendants in~$D$.
This improves the asymptotic running time by a factor of $O(n)$.

{
\newcommand{\FigureA}{./material/examples/chp09_fig040}

\begin{figure}
  \centering

  \begin{tabular}{cc}
  \includegraphics[scale=\scaleExamples]{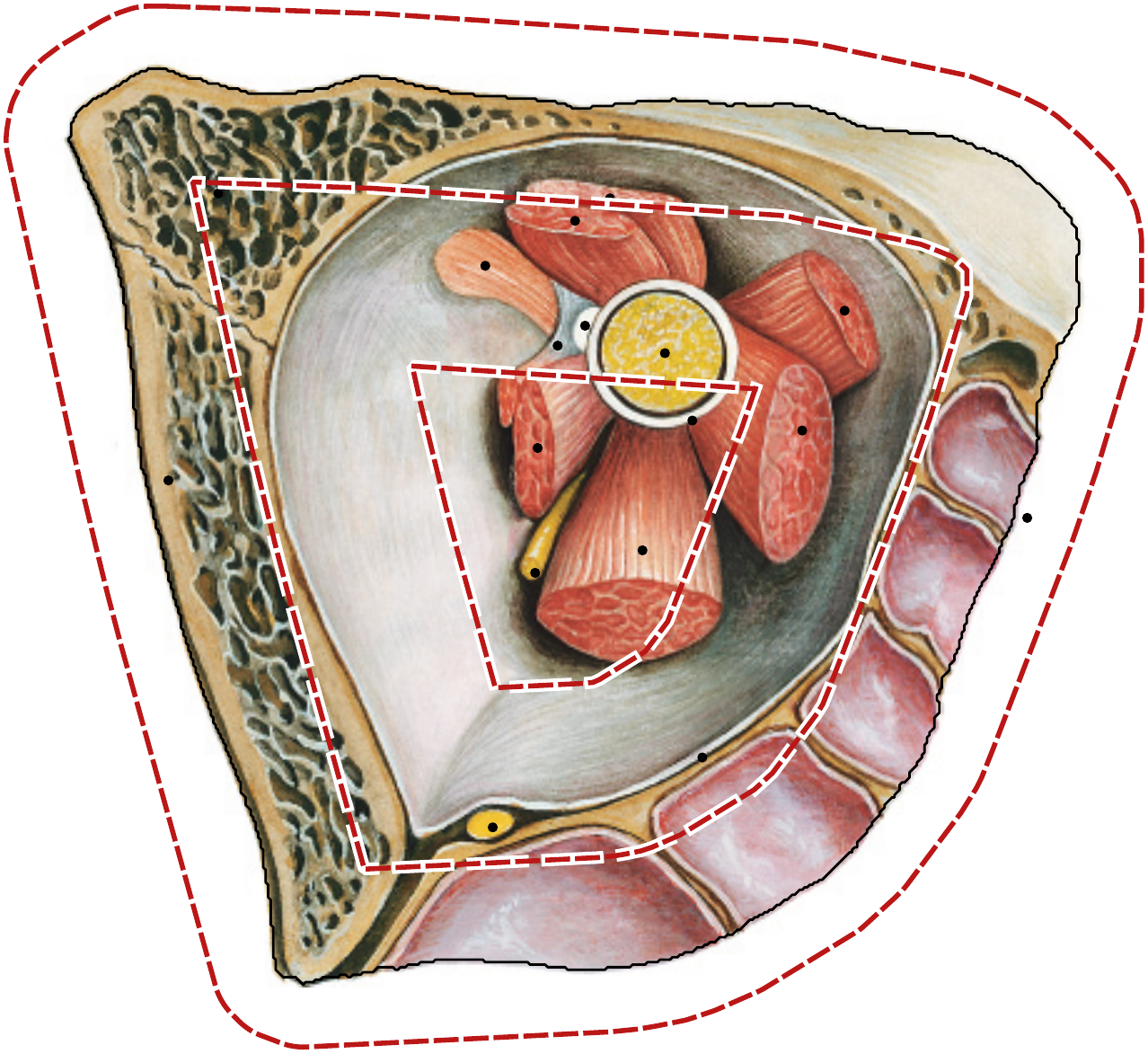}&    
 \includegraphics[scale=\scaleExamples]{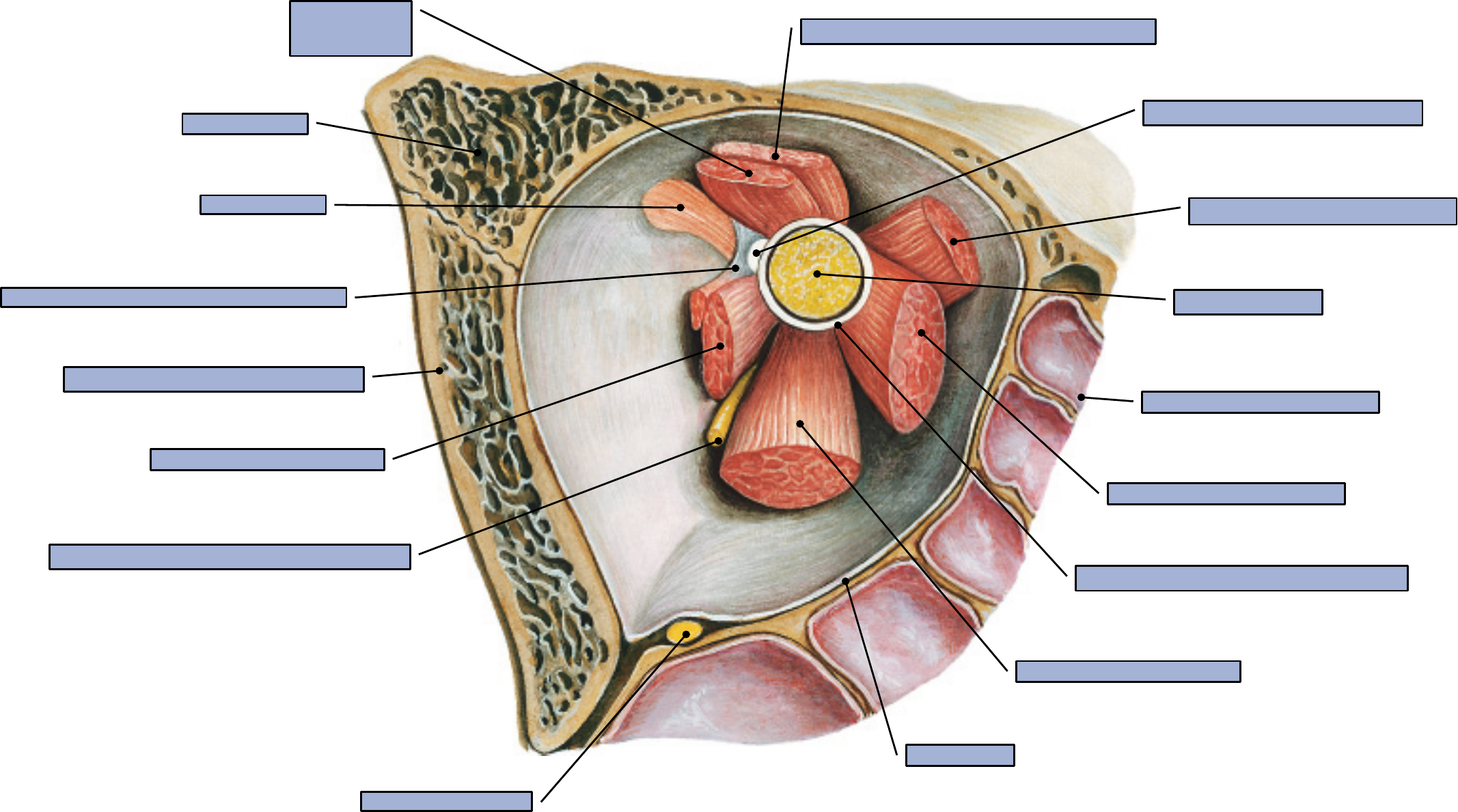}\\
  Contour \& Shells & Original \\ \\
 \includegraphics[scale=\scaleExamples]{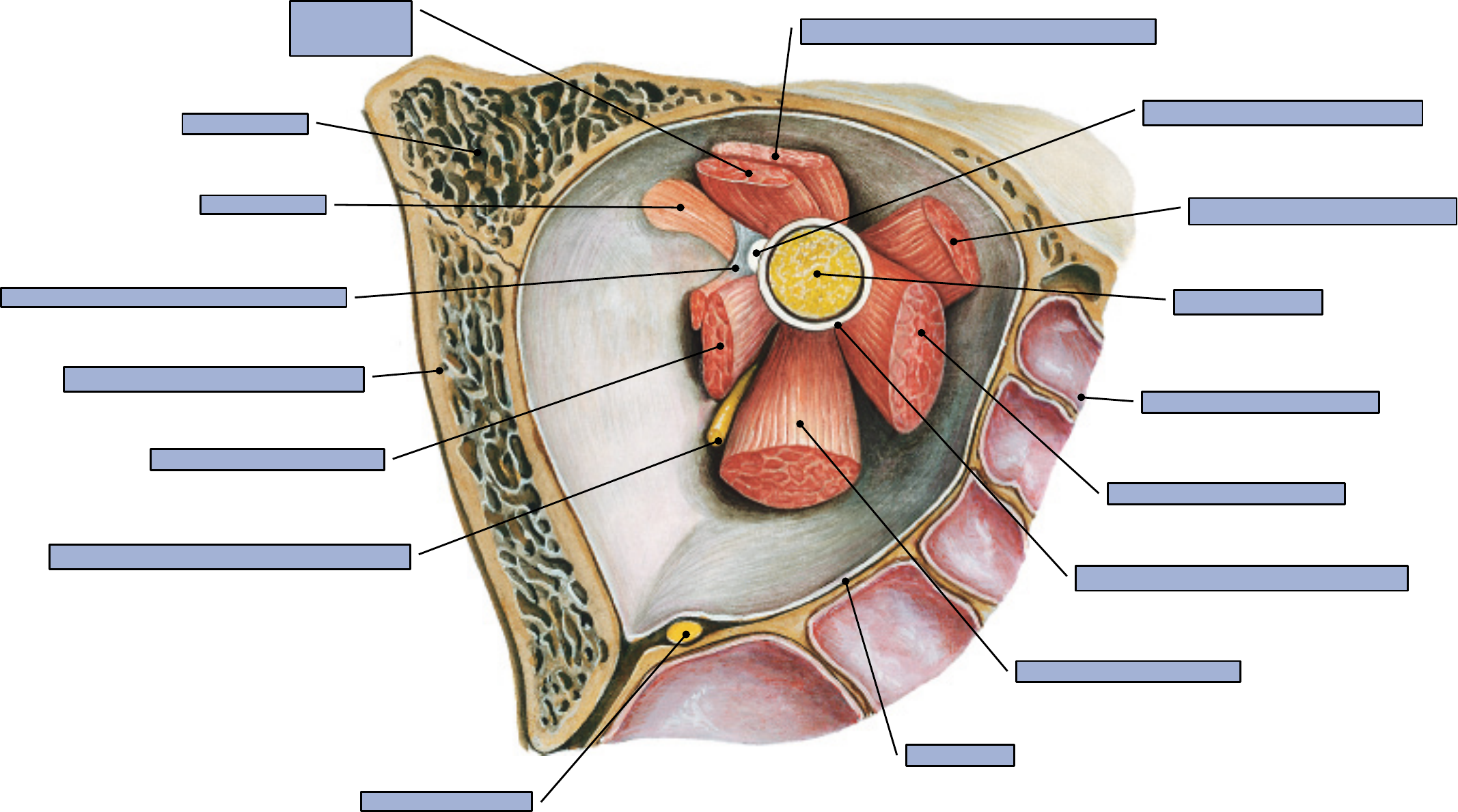}&
 \includegraphics[scale=\scaleExamples]{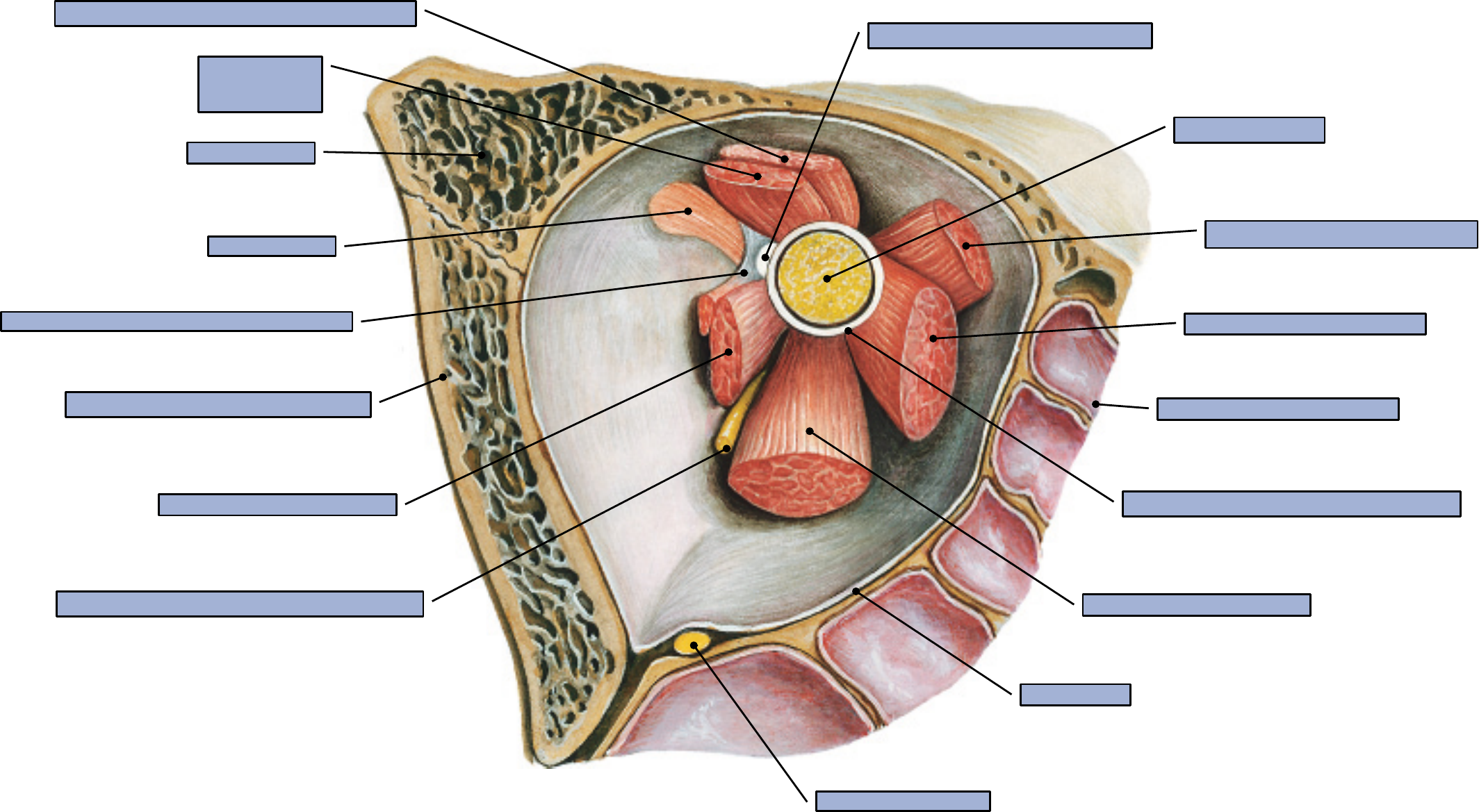}\\
   \OPT(10 sec.), \SCH (2 sec.) & \TSCH (1 sec.) 
 \end{tabular}
 \caption{\OPT, \CH and \SCH produced
   the same labelings. With an cost ratio of $10.2$, the labeling produced by \TSCH is an
   outlier; the labels distances are quite small. 
   Source: Paulsen, Waschke, Sobotta Atlas Anatomie des Menschen,
   23.Auflage 2010 \copyright Elsevier GmbH, Urban \& Fischer,
   München.}
  \label{fig:example:approachesB}
\end{figure}
}

\section{Experimental Evaluation}\label{sec:experiments}
We have implemented a prototype of our approach incorporating the
speed-up techniques of Section~\ref{sec:engineering}.  For simplicity
we constructed the contour of the figure based on its convex hull $H$,
i.e., the contour~$\Contour$ is an exterior offset polygon of $H$ with
distance of $25$ pixels. More sophisticated approaches can be
applied~\cite{Vollick2007}. Similarly, the shells are interior offset
polygons of the contour having distance of $70$ pixels to the next
shell. Both choices are ad-hoc values that mimic handmade drawings,
but a designer may select them depending on the actual
figure. Further, we placed ports by sampling the contour every 10
pixels.

\begin{figure}
 \centering
 \includegraphics{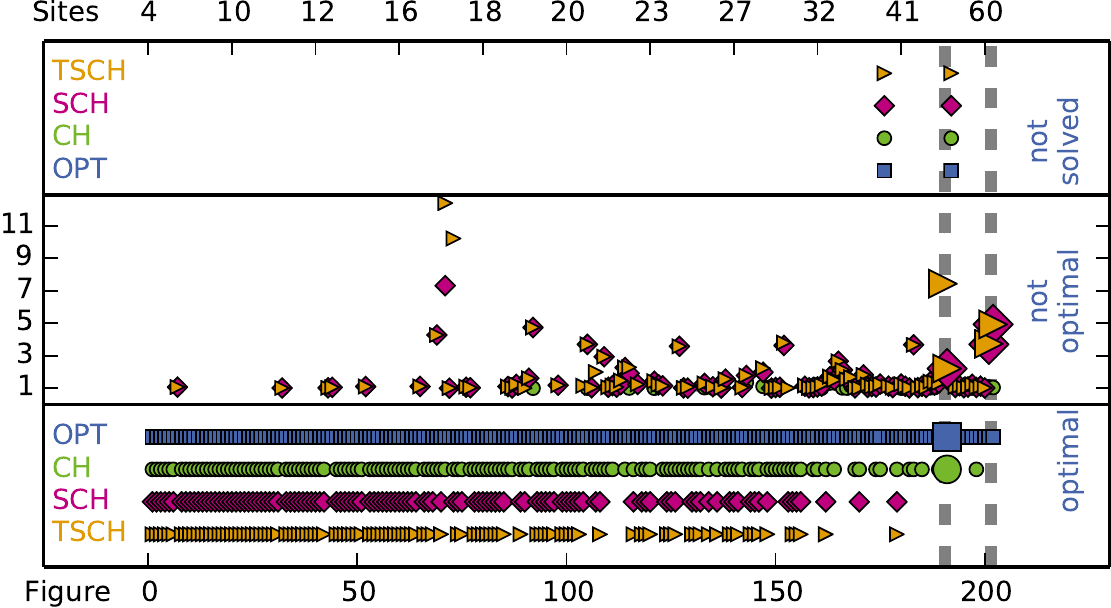}
 \caption{Quality. Each column represents one figure broken down into the labelings computed by the algorithm \OPT (blue rectangle), \CH (green disk), \SCH (pink diamond) or \TSCH (orange triangle). X-axis: The figures are sorted by their number of sites in increasing order. Y-axis: \emph{not solved}: Labelings that could not be constructed. \emph{not optimal:} Cost ratio of non-optimal labelings. \emph{optimal:} Any labeling $\mathcal L_A$ with $c(\mathcal L_A)=c(\mathcal L_\OPT)$. Hereby
$A \in \{\CH,\SCH,\TSCH,\OPT\}$.
   Symbols of labelings violating monotonicity (L\ref{crit:monotone})
   are enlarged and stabbed by a dashed vertical line.}
 \label{fig:quality:pd10}
\end{figure}

The implemented algorithm uses bundles, the speed-up of one-sided
instances, and parallelizes Step 1; see Sect.~\ref{sec:engineering}.
We distinguish the following variants of our algorithm:
\begin{compactenum}
\item \textsc{Optimal (\OPT)}: No further speed-up techniques.
\item \textsc{Capstone-Heuristic (\CH)}: Restricted to capstones.
\item \textsc{Simple-Shell-Heuristic (\SCH)}: Same as \CH, but
  also shells are applied.
\item \textsc{Triangle-Heuristic (\TSCH)}: Same as \SCH, but also including \textsc{Small-Triangles}. 
\end{compactenum}
The implementation was done in C++ and compiled with GCC
4.8.5 using optimization level O2. Further, unless specified otherwise, the experiments were performed on an Intel Xeon
E5-1630v3 processor clocked at 3.7 GHz, with 128 GB of RAM.

We used the 202 extracted medical drawings as input data.  For
reasonably defining the cost functions~$c_1$ and $c_2$, we first
created six labelings for each of five selected figures. These
labelings varied in the choice of minimum label distances, as well as
enforcing monotonicity (L\ref{crit:monotone}) or not. We
discussed these 30 labelings with a domain expert. She confirmed that
monotonicity is an important criterion to obtain balanced labelings
and rated the labelings best where labels with less than 30 pixels
distance were penalized. Smaller and larger penalty thresholds were
rated worse. Further, the expert emphasized that leaders should not
run past sites too closely.  Figures~\ref{fig:example},~\ref{fix:example:approaches} and~\ref{fig:example:approachesB} show three example illustrations
labeled with our algorithms. Further example illustrations are found on \url{http://i11www.iti.kit.edu/contourlabeling/}.

The analysis of handmade labelings also supports monotonicity:
$93.8\%$ of the labels satisfy this property. For about $70\%$ of the
test instances violating monotonicity, the violation was less than
$10^\circ$ in maximum. About $93\%$ of these instances have at most
$5$ violations.

We incorporated these findings into the cost function as follows. Let
$M$ be a big constant. Any label candidate and any instance with cost
at least $M$ is excluded. In our experiments we set $M=10^9$.  The
cost function $c_1$ takes the leader's length and the smallest
distance to sites into account. More precisely, any candidate label
whose leader is three times longer than the shortest possible leader
for the same site is excluded.  Hence, sites located close to the
contour of the figure have short leaders, while sites in the center of
the figure are not affected by this exclusion at all. Let $\ell$ be a
label candidate and $\lambda$ be the leader of $\ell$.  If the
distance~$d$ of $\lambda$ to any site (not connected to $\lambda$) is
less than 10 pixels, we set $c_1(\ell)=\length(\lambda)^2 +
M/(100\cdot d)$ and otherwise $c_1(\ell)=\length(\lambda)^2$. Hence,
we penalize both long leaders as well leaders that closely run past a
site.

We define $c_2$ as follows.
Let $\ell_1$ and $\ell_2$ be a pair of possible consecutive label
candidates. We set $c_2(\ell_1,\ell_2)=M$ if $\ell_1$ and $\ell_2$
violate monotonicity by more than $10^\circ$, excluding these
pairs. For smaller violations we set $c_2(\ell_1,\ell_2)=M/6+c_v$,
which effectively allows at most $5$ of these violations in total. If
$\ell_1$ and $\ell_2$ satisfy monotonicity, we set
$c_2(\ell_1,\ell_2)=c_v$. Here, $c_v$ is the cost caused by the
vertical distance $d_v$ of $\ell_1$'s and $\ell_2$'s text boxes:  If
$\ell_1$ and $\ell_2$ lie on different sides of $\Contour$, we
set~$c_v=0$. Otherwise, if the vertical distance~$d_v$ is less than $5$ pixels, we set
$c_v=M$ excluding these pairs. If $5\leq d_v< 30$ pixels, we set
$c_v=M/(100\cdot d_v)$ penalizing too small distances, and in all other
cases $c_v=0$.

\begin{figure}[t]
  \centering
 \includegraphics[]{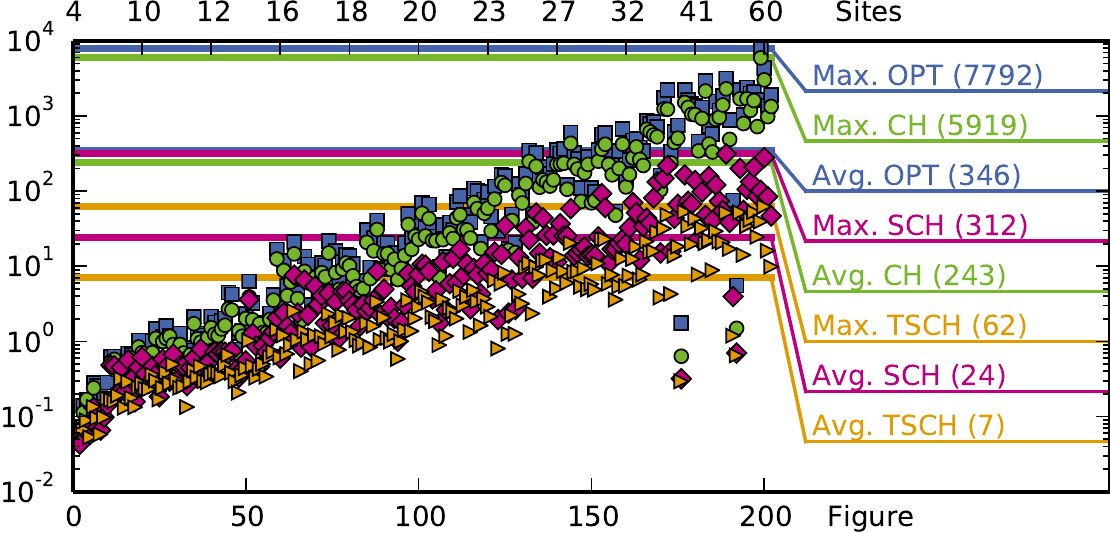}
 \caption{Running time in seconds (log.\ scale).  Each column represents one figure broken down into the labelings computed by the algorithm \OPT (blue rectangle), \CH (green disk), \SCH (pink diamond) or \TSCH (orange triangle). X-axis: The figures are sorted by their number of sites in increasing order. }
 \label{fig:running-time:pd10}
\end{figure}

\textbf{Quality.} To analyze the quality of the labelings constructed
by \CH, \SCH\ and \TSCH, we compare the \emph{cost ratio} $c(\mathcal
L_A)/c(\mathcal L_\OPT)$, where $\mathcal L_\OPT$ is created by $\OPT$
and $\mathcal L_A$ is created by the variant $A\in
\{\CH,\SCH,\TSCH\}$; see Fig.~\ref{fig:quality:pd10}. About $73 \%$
(\CH), $56\%$ (\SCH) and $53\%$ (\TSCH) of the labelings achieve
optimal costs. For $90\%$ of the figures, the algorithms achieve
labelings whose costs are at most a factor of $1.06$ (\CH), $1.75$
(\SCH) and $1.99$ (\TSCH) worse than the optimal costs.  Only for two
figures no valid solution could be computed, because their contour was
too small to host all labels; see Fig.~\ref{fig:bad-instances}. In
particular, the bottom sides of the contours are almost horizontal,
which means that too many labels must be placed along the left and
right side of the contour.   The designer of the original
  labeling avoided this problem by breaking design rule T5.

For the majority of the figures monotone labelings (criteria
L\ref{crit:monotone}) were created. Only for one figure none of the
algorithms could create a monotone labeling. For two further figures
$\SCH$ and $\TSCH$ could not create monotone labelings, while the
other two approaches did. Finally, for one figure only $\TSCH$ did not
create a monotone figure.

A consulted domain expert stated that the created labelings would be
highly useful as initial labeling for the remaining process of laying
out the figure. According to the expert, the labelings already
have high quality and would require only minor changes, due to
aesthetic reasons. These can be hardly expressed as general criteria,
but rely on the expertise of the designer. Altogether, the domain expert
assessed our algorithm to be a tool of great use that could reduce the
working load of a designer significantly.

\textbf{Running Time.} The average running times of our algorithms
range between $7$ seconds (\TSCH) and $346$ seconds (\OPT); see
Fig.~\ref{fig:quality:pd10}.  The variants \SCH and \TSCH are
remarkably faster than \OPT; see Fig.~\ref{fig:quality:pd10}. On
average they achieve a speedup by a factor of $8.4$ and $23.0$,
respectively. For some figures, \TSCH and \SCH even achieve a speed up
of $198$ and $76$. Further, \TSCH and \SCH never exceeded $62$ and
$312$ seconds, respectively. On average \TSCH is by a factor $2.5$ faster than \SCH;
in maximum by a factor of $7.1$.  The variant \CH only slightly
improves \OPT by a factor of $1.4$ on average and $3.7$ in maximum.

Since \OPT has a high memory consumption, we ran the experiments on a
server with 128 GB RAM. When such a system is not available, \SCH and
\TSCH are appropriate alternatives for \OPT, because they use
significantly less memory, are fast and mostly produce labelings of
high quality. To assess the applicability of the approaches in a
typical setting, we ran both \SCH and \TSCH on an ordinary laptop with
an Intel Core i7-3520M CPU clocked at $2.9$ GHz and $8$ GB of RAM. In
comparison to the previous setting, \SCH and \TSCH are slower by a
factor of $1.24$ and $1.22$ in maximum, respectively. On average \SCH
needs $28$ seconds and \TSCH needs $8$ seconds. Within $27$ ($94$)
minutes the labelings of all 202 figures were produced by \TSCH
(\SCH); in contrast, a domain expert stated that creating a labeling
for a figure with about $50$ sites by hand
 may easily take $30$ minutes.

\begin{figure}[t]
  
  \includegraphics[width=0.45\linewidth]{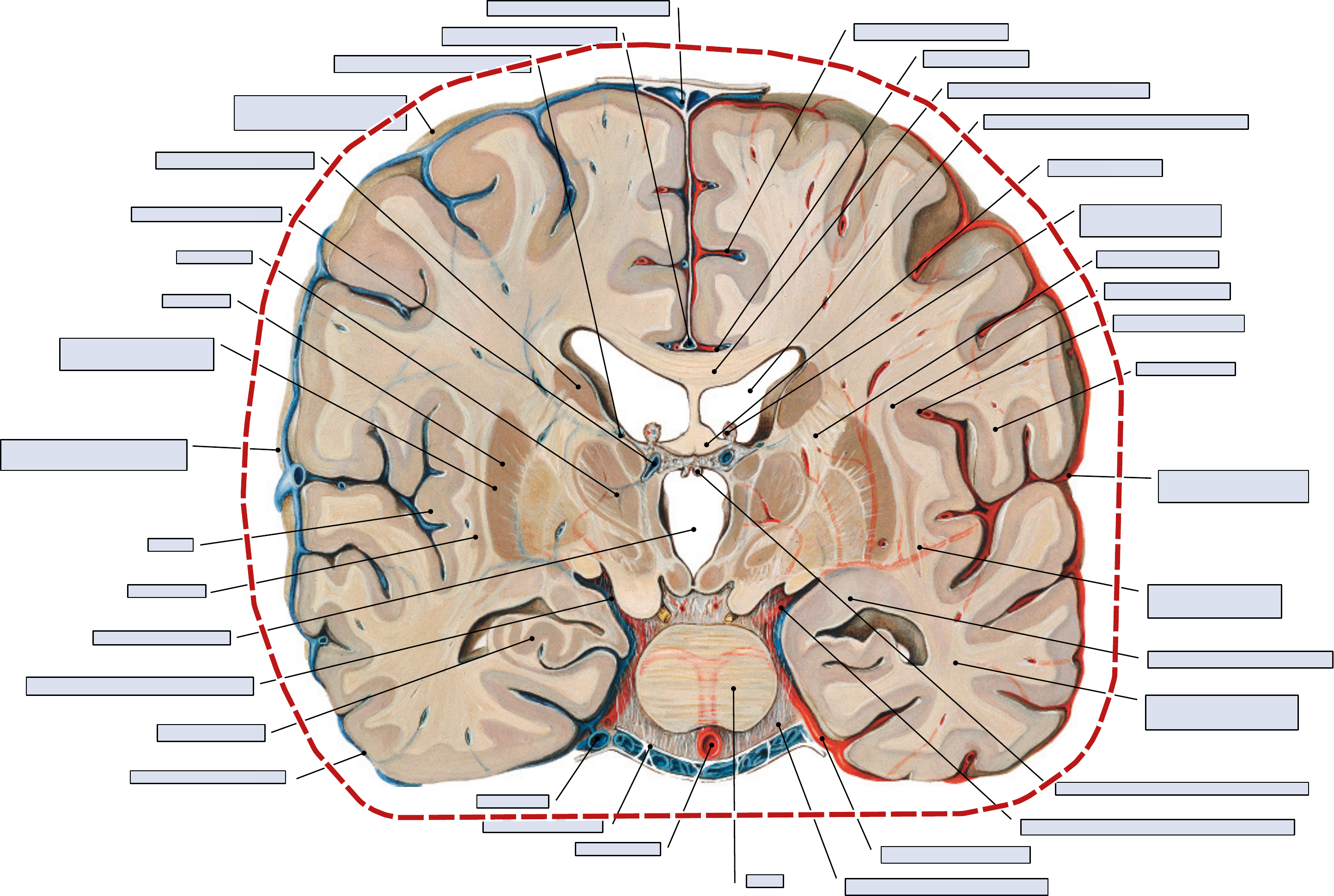}  
 \hfill 
 \includegraphics[width=0.45\linewidth]{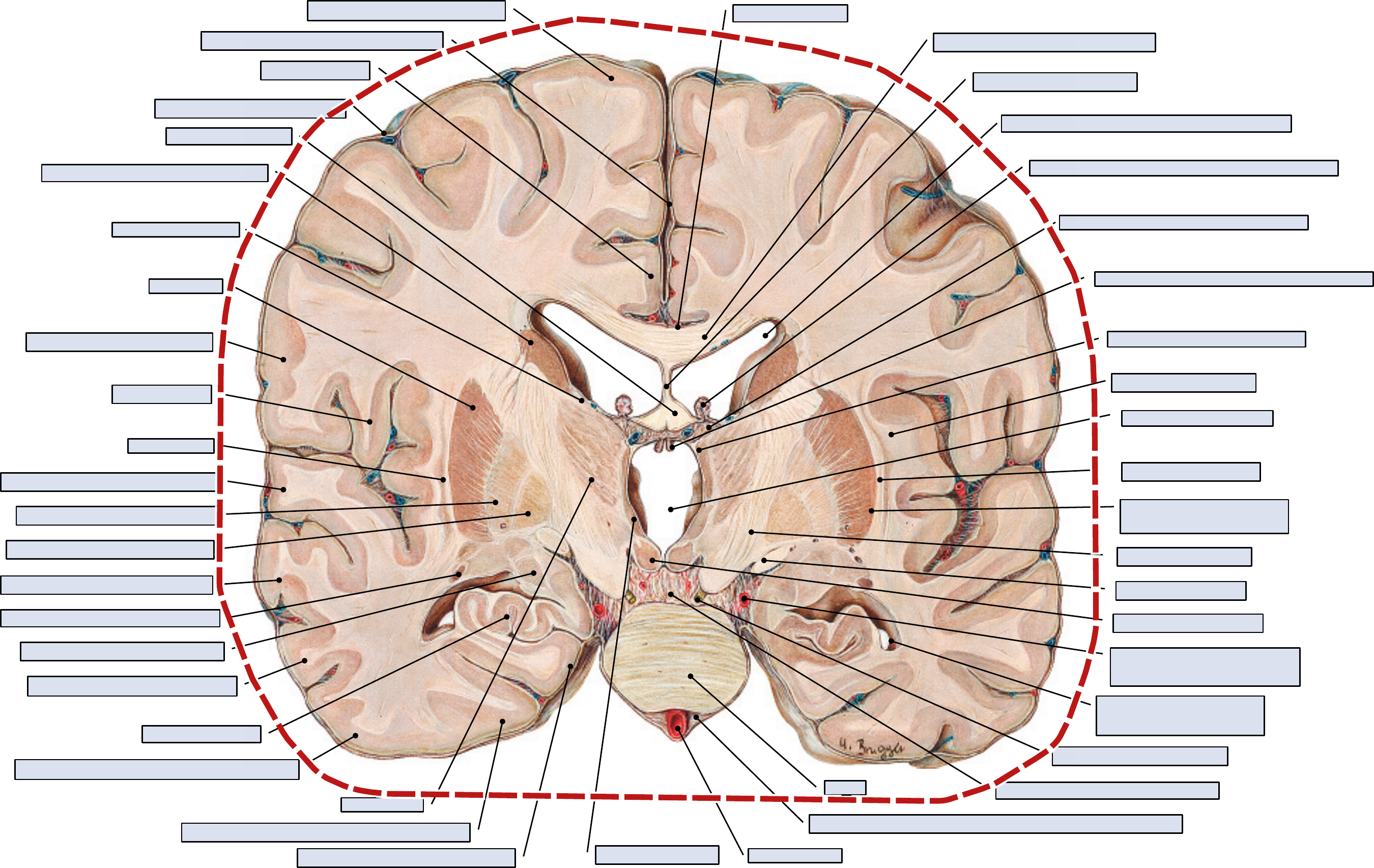}

 \caption{Illustrations with original labelings.  Due to the choice of
   the contour (red), both instances could not be solved by our
   approach. Source: Paulsen, Waschke, Sobotta Atlas Anatomie des
   Menschen, 23.Auflage 2010 \copyright Elsevier GmbH, Urban \&
   Fischer, München.}
  \label{fig:bad-instances}
\end{figure}

\section{Discussion}

The evaluation shows that our approach computes high-quality labelings
for the vast majority of instances in short time.  While \OPT and \CH
are mainly useful for evaluating our approach, \SCH and \TSCH are fast
enough to be deployed in practice. If quality is more important than
running time, \SCH is preferable and otherwise \TSCH%
Figure~\ref{fig:example:approachesB} shows an example where \SCH
yields a better result than \TSCH, still in reasonable time.

Only for two illustrations we could not create any labelings due to
the chosen contour. To avoid this problem one could incorporate more
sophisticated procedures for creating the contour.  However, we
refrained from this, because we were mainly interested in evaluating
the performance of our dynamic programming approach. Alternatively,
the designer could appropriately adapt the contour in an interaction
step.

Further, one can relax the convexity assumption in practice
allowing more general contours. The cases when overlaps of text boxes
may occur for non-convex contours are pathological.  This further
enhances the possibility to integrate interaction with the user, who
could adapt the contour on demand.  In most cases the running times of
\TSCH are sufficient for interactive editing, which is necessary for
laying out professional books; for real-time scenarios (e.g., labeling
images of ongoing surgeries) the performance of our approach needs 
 further improvement. Identifying more rules for excluding unnecessary
instances is one possibility to speed up the procedure.

Although the produced labelings already have high quality, we can
improve on them by applying post-processing steps, e.g., fix the order
of the labels and restart the dynamic programming approach on a larger
set of ports to do fine-tuning on the label placement.  More
sophisticated post-processing steps are future work.

Similarly, for labeling the same image in different scales, one could,
based on a large-scale master labeling, penalize changes in the radial
ordering of the labels in other scales. Pre-computing the labelings
for all given scales in that way, one minimizes layout changes between
different zoom levels, which can be useful for interactive views. A
detailed evaluation %
is also future work.

Finally, our approach supports a one-to-one correspondence between
labels and sites. To also support multiple sites per label, one could
create multiple leader candidates having different fork
points. Excluding them mutually, this yields a simple adaption of our
original approach. The main research questions is then to identify
appropriate candidate positions for the fork points.

\section{Conclusion}
In this paper we presented a flexible model for contour
labeling, which we validated through interviews with domain
experts and a semi-automatic analysis on handmade labelings.  With
some engineering the developed dynamic programming approach can be
used to generate labelings of high quality in short time. The
presented approach is particularly interesting for creating labelings
of large collections of figures that must follow the same design
rules; a prominent example are figures in atlases of human anatomy.

In contrast to external labeling heuristics
(e.g.,~\cite{Hartmann2004,Ali2005}), our full approach provides
mathematically optimal, exact solutions. Moreover, we can use these
optimal solutions to assess the performance of our faster heuristics
quantitatively in terms of the cost ratio. Compared to other exact
labeling algorithms (e.g.~\cite{Bekos2007,Benkert2009}), our
optimization model supports more general contours, non-uniform label
shapes, and the flexible integration of additional soft and hard
constraints. %

\end{document}